\documentclass{vldb}

\newif\ifvldb
\newif\ifsubqueries

\usepackage{times}
\usepackage[hyphens]{url}
\usepackage{cite}
\usepackage{color}
\usepackage{graphicx}
\graphicspath{ {figures/} }

\usepackage[table]{xcolor}
\usepackage{outlines}
\usepackage{multirow}
\usepackage{xspace}
\usepackage{color,colortbl}
\usepackage{amssymb}
\usepackage{amsmath}
\usepackage[most]{tcolorbox}
\usepackage{marginnote}

\usepackage{amsthm}
\usepackage{csquotes}

\usepackage{float}
\usepackage{listings}
\usepackage{inconsolata}
\usepackage{pifont}
\usepackage[font=bf,skip=3pt]{caption}

\usepackage{enumitem}
\setlist{nosep}
\ifvldb
\usepackage{flushend}
\else
\fi

\usepackage{tikz}
\usetikzlibrary{positioning}

\definecolor{midnightblue}{rgb}{0.1, 0.1, 0.44}
\definecolor{darkslateblue}{rgb}{0.28, 0.24, 0.55}
\definecolor{smokyblack}{rgb}{0.06, 0.05, 0.03}
\definecolor{xanadu}{rgb}{0.25, 0.23, 0.27}
\renewcommand{\paragraph}[1]{\vspace*{0.5mm}\noindent\textbf{#1.}}
\newcommand{\case}[1]{\vspace*{1mm}\noindent\textbf{#1.}}

\newcommand{\flex}{{\sc Flex}\xspace}
\newcommand{\flexlarge}{{\large \flex}\xspace}

\newcommand{\allegro}{{\sc Allegro}\xspace}
\newcommand{\numqueries}{8.1 million\xspace}
\newcommand{\numqueriesrounded}{8.1 million\xspace}
\newcommand{\note}[1]{{\color{gray}\underline{Note:} \textit{#1}}}
\newcommand{\todo}[1]{{\color{red}\textbf{TODO: #1}}}

\renewcommand{\note}[1]{}
\renewcommand{\todo}[1]{}

\newcommand{\select}{\texttt{Select}\xspace}
\newcommand{\join}{\texttt{Join}\xspace}

\newcommand{\staticsensitivity}{elastic sensitivity\xspace}
\newcommand{\Staticsensitivity}{Elastic sensitivity\xspace}
\newcommand{\StaticSensitivity}{Elastic Sensitivity\xspace}
\newcommand{\vi}{\textsf{vr}}
\newcommand{\mf}{\textsf{mf}}
\newcommand{\sfun}{\hat{S}}
\newcommand{\dsfun}{\hat{\Delta S}}

\newtheorem{theorem}{Theorem}
\newtheorem{lemma}{Lemma}
\newtheorem{definition}{Definition}

\newcommand{\argmin}{\operatornamewithlimits{argmin}}

\newcommand{\rmspc}{\vspace*{-3mm}}
\newcommand{\rmspcsm}{\vspace*{-1mm}}

\usepackage{xpatch}
\makeatletter
\xpatchcmd{\endmdframed}
  {\aftergroup\endmdf@trivlist\color@endgroup}
  {\endmdf@trivlist\color@endgroup\@doendpe}
  {}{}
\makeatother
\makeatletter
\g@addto@macro\normalsize{%
  \setlength\abovedisplayskip{-5pt}
  \setlength\belowdisplayskip{4pt}
  \setlength\abovedisplayshortskip{-5pt}
  \setlength\belowdisplayshortskip{4pt}
}
\let\OLDthebibliography\thebibliography
\renewcommand\thebibliography[1]{
  \OLDthebibliography{#1}
  \setlength{\parskip}{0pt}
  \setlength{\itemsep}{1pt plus 0.3ex}
}

\vldbfalse
\subqueriestrue

\begin{document}

\vldbTitle{Towards Practical Differential Privacy for SQL Queries}
\vldbAuthors{Noah Johnson, Joseph P. Near, Dawn Song}
\vldbVolume{11}
\vldbNumber{5}
\vldbYear{2018}
\vldbDOI{https://doi.org/10.1145/3177732.3177733}

\ifvldb
\else
\toappear{..}
\fi

\title{Towards Practical Differential Privacy for SQL Queries}

\numberofauthors{3}
\author{
\alignauthor
Noah Johnson\\
\affaddr{University of California, Berkeley}\\
\email{noahj@berkeley.edu}
\alignauthor
Joseph P. Near\\
\affaddr{University of California, Berkeley}\\
\email{jnear@berkeley.edu}
\alignauthor
Dawn Song \\
\affaddr{University of California, Berkeley}\\
\email{dawnsong@cs.berkeley.edu}
}

\maketitle

\begin{abstract}
  Differential privacy promises to enable general data analytics while
  protecting individual privacy, but existing differential privacy
  mechanisms do not support the wide variety of features and databases
  used in real-world SQL-based analytics systems.

  This paper presents the first practical approach for differential
  privacy of SQL queries. Using \numqueries real-world queries, we conduct
  an empirical study to determine the requirements for practical differential privacy,
  and discuss limitations of previous approaches in light of these requirements.
  To meet these requirements we propose
  \staticsensitivity, a novel method for approximating the local
  sensitivity of queries with general equijoins. We prove that elastic sensitivity
  is an upper bound on local sensitivity and can therefore be used to
  enforce differential privacy using any local sensitivity-based mechanism.
  
  We build \flex, a practical end-to-end system to enforce differential privacy
  for SQL queries using elastic sensitivity.
  We demonstrate that FLEX is compatible with any existing database,
  can enforce differential privacy for real-world SQL queries,
  and incurs negligible (0.03\%) performance overhead.
\end{abstract}

\section{Introduction}

As organizations increasingly collect sensitive information about
individuals, these organizations are ethically and legally obligated
to safeguard against privacy leaks. Data analysts within these
organizations, however, have come to depend on unrestricted access to
data for maximum productivity. This access is frequently provided in
the form of a relational database that supports SQL queries.
Current approaches for data security and privacy cannot guarantee privacy 
for individuals while providing general-purpose access for the analyst.

As demonstrated by recent insider attacks~\cite{uci_medical,
  cpmc_medical, morganstanley, swiss}, allowing members of an
organization unrestricted access to data is a major cause of privacy
breaches. Access control policies can limit access to a particular
database, but once an analyst has access, these policies cannot
control how the data is used.
Data anonymization attempts to provide privacy while allowing
general-purpose analysis, but cannot be relied upon, as demonstrated
by a number of re-identification
attacks~\cite{sweeney1997weaving,DBLP:journals/corr/abs-cs-0610105,taxis,de2013unique}.

Differential privacy~\cite{dworkdifferential2006,
  dwork2006calibrating, dwork2014algorithmic} is a promising technique
  for addressing these issues. Differential privacy allows general statistical analysis of data while
protecting \emph{data about individuals} with a strong formal
guarantee of privacy.

Because of its desirable formal guarantees, differential privacy has received
growing attention from organizations including
Google and Apple. However, research on practical
techniques for differential privacy has focused on special-purpose use cases,
such as collecting statistics about web browsing behaviors~\cite{erlingsson2014rappor} 
and keyboard and emoji use~\cite{apple}, while differential
privacy for general-purpose data analytics remains an open challenge.

Various mechanisms~\cite{mcsherry2009privacy, proserpio2014calibrating,
  nissim2007smooth, mohan2012gupt, Blocki:2013:DPD:2422436.2422449,
  narayan2012djoin} provide differential privacy for some subsets of
SQL-like queries but none support the majority of queries in
practice.
These mechanisms also require modifications to the
database engine, complicating adoption in practice.

Furthermore, although the theoretical aspects of differential privacy
have been studied extensively, little is known about the quantitative
impact of differential privacy on real-world queries. Recent work has
evaluated differential privacy mechanisms on real-world data~\cite{DBLP:conf/ndss/BlockiDB16, DBLP:conf/sigmod/HayMMCZ16,
  DBLP:journals/pvldb/HuYYDCYGZ15}, however this work uses a
  limited set of queries representing a single,
special-purpose analytics task such as histogram
analysis~\cite{DBLP:conf/ndss/BlockiDB16} or range queries~\cite{DBLP:conf/sigmod/HayMMCZ16}.
To the best of our knowledge,
no existing work has explored the design and evaluation of 
differential privacy techniques for general, real-world queries.

This paper proposes \emph{elastic sensitivity}, a novel approach for differential
privacy of SQL queries.
In contrast to existing work, our approach is compatible with real
database systems, supports queries expressed in standard SQL,
and integrates easily into existing data environments. The work
therefore represents a first step towards practical differential privacy.
The approach has recently been adopted by Uber to enforce differential
privacy for internal data analytics~\cite{uber_blog_post}.

We developed elastic sensitivity using requirements derived from a
dataset of \numqueriesrounded real-world queries.
This paper focuses on counting queries, which constitute the majority of
statistical queries in this dataset, and discusses extensions to
the approach for other aggregation functions. We have released an open-source tool for computing elastic sensitivity of SQL queries~\cite{github_repo}.

\paragraph{Contributions} We make four primary contributions toward
practical differential privacy:

\begin{enumerate}[topsep=1mm,leftmargin=4mm]
\itemsep1.0mm
\item We conduct the largest known empirical study
  of real-world SQL queries---\numqueries queries in total. 
  From these results we show that the queries used in prior work to
  evaluate differential privacy mechanisms are not representative
  of real-world queries. We propose a new set of
  requirements for practical differential privacy on SQL queries
  based on these results.

\item To meet these requirements, we propose
  \emph{\staticsensitivity}, a sound approximation of local sensitivity~\cite{dwork2009differential,nissim2007smooth} that supports general equijoins and can be calculated
  efficiently using only the query itself and a set of precomputed database metrics. We prove
  that elastic sensitivity is an upper bound on local sensitivity and can therefore be
  used to enforce differential privacy using any local sensitivity-based mechanism.

\item We design and implement \flex, an end-to-end differential privacy system for SQL queries
 based on elastic sensitivity. We demonstrate that \flex is compatible with \emph{any}
  existing database, can enforce differential privacy for the majority of real-world SQL queries,
  and incurs negligible (0.03\%) performance overhead.

\item In the first experimental evaluation of its kind, we use \flex 
  to evaluate the impact of differential privacy on 9862 real-world statistical queries in our dataset.
  In contrast to previous empirical evaluations of differential privacy,
  our experimental set contains a diverse variety of real-world queries executed on real data.
  We show that \flex introduces low error for a majority of these queries.
\end{enumerate}

The rest of the paper is organized as
follows. Section~\ref{sec:survey-data-analysis} contains our empirical
study and defines the requirements for a practical differential
privacy mechanism.
In Section~\ref{sec:static-sensitivity}, we define \staticsensitivity
and prove that elastic sensitivity is an upper bound on local sensitivity. In Section~\ref{sec:flex} we
describe \flex, our system for enforcing differential privacy using elastic sensitivity.
Section~\ref{sec:exper-eval} contains our experimental evaluation of \flex and 
Section~\ref{sec:related-work} surveys related work.

\section{Requirements for Practical \\Differential Privacy}
\label{sec:survey-data-analysis}

\def\questionone{How many different database backends are used?}
\def\questiontwo{Which relational operators are used most frequently?}
\def\questionthreea{What types of joins are used most frequently?}
\def\questionthreeb{How many joins are used by a typical query?}
\def\questionfour{What fraction of queries use aggregations?}
\def\questionfive{Which aggregation functions are used most frequently?}
\def\questionsixa{How complex are typical queries?}
\def\questionsixb{How large are typical query results?}

We use a dataset consisting of millions of
SQL queries to establish requirements for a practical differential privacy
system that supports the majority of real-world queries. We
investigate the limitations of existing general-purpose differential mechanisms in light of these requirements, and introduce \emph{\staticsensitivity}, our new approach that meets these requirements.

\ifvldb

\else
We investigate the following properties of queries in our dataset:

\begin{itemize}[topsep=1mm,leftmargin=4mm]
\itemsep0.5mm

\item \textbf{\questionone}
  A practical differential privacy system must integrate with existing
  database infrastructure.

\item \textbf{\questiontwo}
  A practical differential privacy system must at a minimum support the most common
  relational operators.

\item \textbf{What types of joins are used most frequently and many are used by a typical query?}
  Making joins differentially private is challenging because the output
  of a join may contain duplicates of sensitive rows. This duplication is
  difficult to bound as it depends on the join type, join condition, and the
  underlying data. Understanding the different types of joins and their relative
  frequencies is therefore critical for supporting differential privacy on
  these queries.

\item \textbf{What fraction of queries use aggregations and which aggregation functions are used most frequently?}
  Aggregation functions in SQL return statistics about populations in the data.
  Aggregation and non-aggregation queries represent fundamentally different privacy problems, as will be shown.
  A practical system must at minimum support the most common aggregations.

\item \textbf{How complex are typical queries and how large are typical query results?}
  To be practical, a differential privacy mechanism must support real-world queries
  without imposing restrictions on query syntax, and it must scale to typical result sizes.
\end{itemize}

\fi

\ifvldb
\else

\paragraph{Dataset}
We use a dataset of SQL queries written by employees at Uber. The
dataset contains \numqueries queries executed between March 2013 and
August 2016 on a broad range of sensitive data including rider and
driver information, trip logs, and customer support data.

Data analysts at Uber query this information in support of many
business interests such as improving service, detecting fraud, and
understanding trends in the business. The majority of these use-cases
require flexible, general-purpose analytics.

Given the size and diversity of our dataset, we believe it is
representative of SQL queries in other real-world situations.
\fi

\newcounter{question}
\newcommand{\question}[1]{\stepcounter{question}\vspace{3mm}\noindent\textbf{\textbf{Question
    \thequestion:} #1}\nopagebreak}

\newcommand{\results}{\noindent\emph{\textbf{Results.}}\xspace}

\ifvldb

\else

\subsection{Study Results}
\label{sec:study-results}

We first summarize the study results, then define requirements
of a practical differential privacy technique for real-world queries
based on these results.

\question{\questionone}
\rmspc
\begin{figure}[H]
  \centering
  \includegraphics[width=0.45\textwidth]{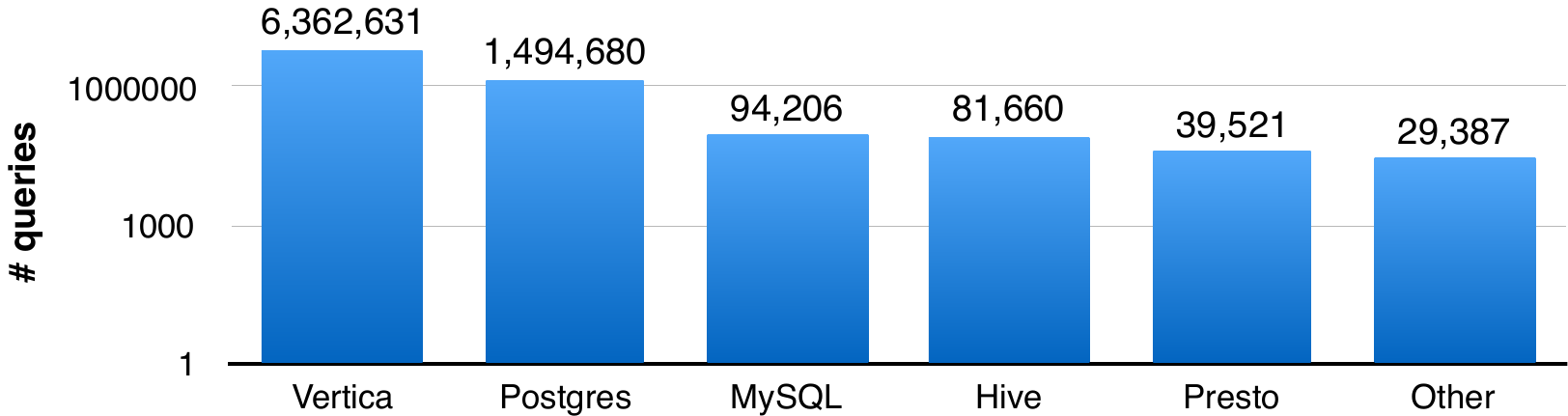}
\end{figure}
\rmspc
\results The queries in our dataset use more than 6 database backends,
including Vertica, Postgres, MySQL, Hive, and Presto.

\question{\questiontwo}
\rmspc
\begin{table}[H]
    \begin{tabular}{c l}
      {\hspace*{-3mm}\small
  \begin{tabular}{l r}
    \textbf{Operator} & \textbf{Frequency} \\
  		\hline
    \texttt{Select}        & 100\% \\
    \texttt{Join}          & 62.1\% \\
    \texttt{Union}         & 0.57\% \\
    \texttt{Minus/Except}  & 0.06\% \\
    \texttt{Intersect}     & 0.03\% \\
  \end{tabular}}

    &
      \begin{minipage}{.25\textwidth}
        \vspace*{1mm}
\results All queries in our dataset use the \select operator, more than half of the queries use the \join operator, and fewer than 1 percent use other operators such as \texttt{Union}, \texttt{Minus}, and \texttt{Intersect}.
      \end{minipage}
  \end{tabular}
\end{table}
\rmspc

\question{\questionthreeb}
\rmspc
\begin{figure}[H]
  \centering
\includegraphics[width=.40\textwidth]{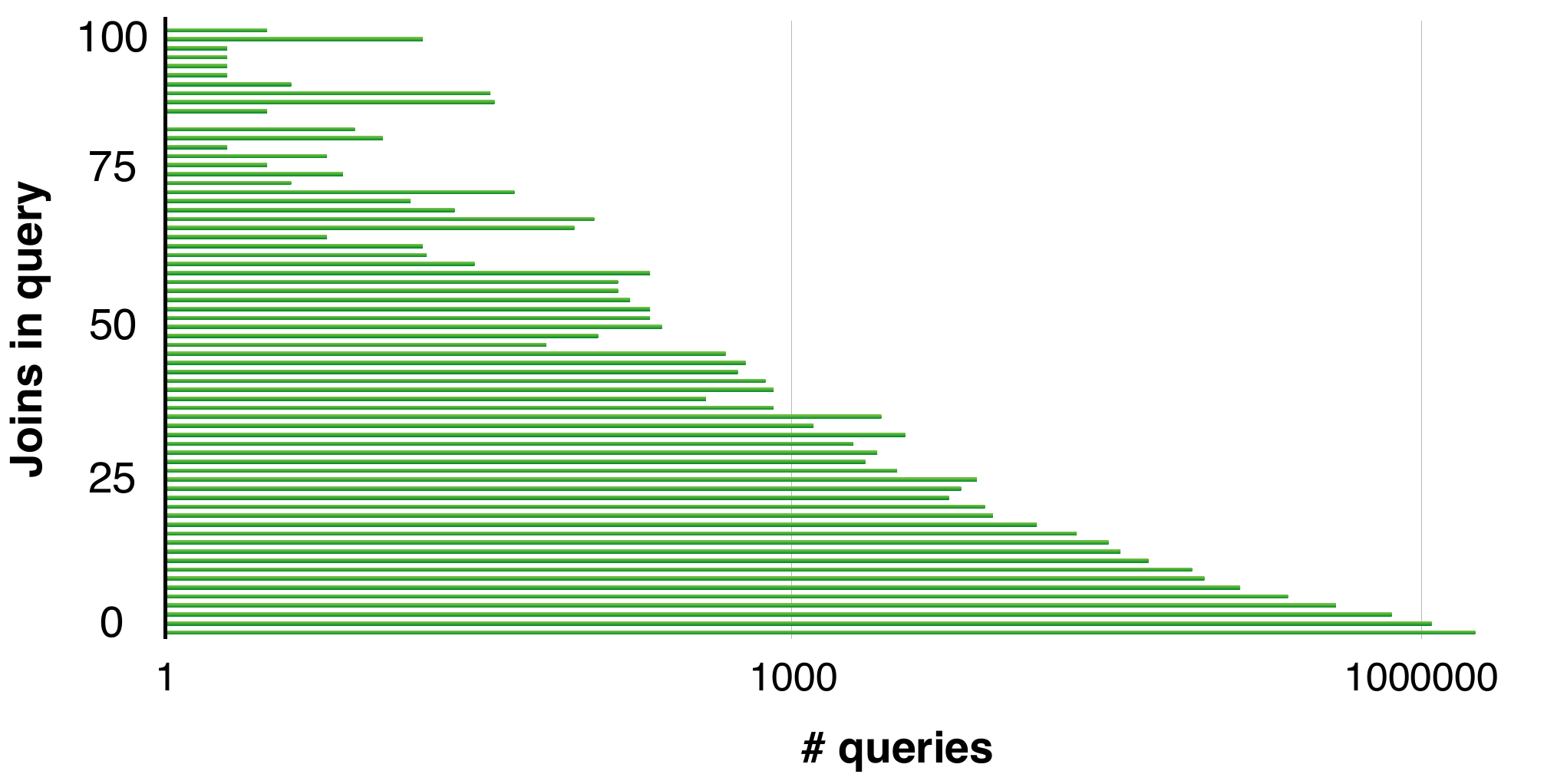}  
\end{figure}

\rmspc
\results A significant number of queries use multiple joins, with queries using as many as 95 joins.

\question{\questionthreea}

\rmspcsm
\begin{figure}[H]
  \centering
\includegraphics[width=.47\textwidth]{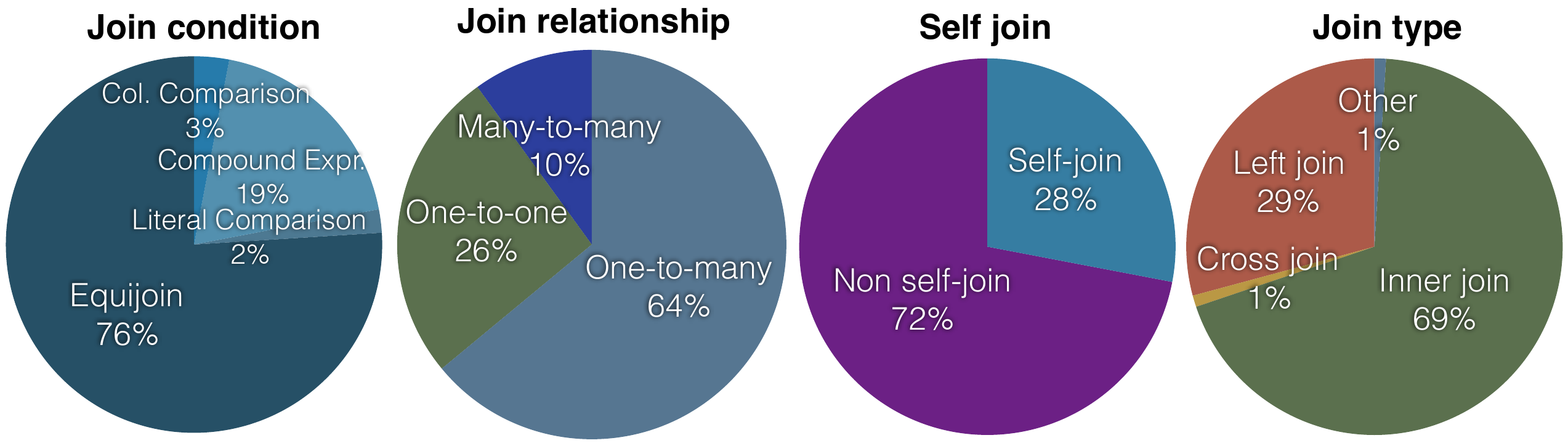}  
\end{figure}
\rmspcsm

\noindent\emph{\textbf{Join condition.}}
The vast majority (76\%) of joins are \emph{equijoins}: joins that are conditioned on value equality of one column from both relations. A separate experiment (not shown) reveals that 65.9\% of all join queries use \emph{exclusively}
equijoins. 

Compound expressions, defined as join conditions using function
applications and conjunctions and disjunctions of primitive operators, account for 19\% of
join conditions. Column comparisons, defined as conditions that compare two
columns using non-equality operators such as \emph{greater
than}, comprise 3\% of join conditions. Literal comparisons, defined as join conditions comparing a single column to a string or integer literal, comprise 2\% of join conditions.

\vspace{1mm}
\noindent\emph{\textbf{Join relationship.}}
A majority of joins (64\%) are conditioned on one-to-many relationships, over one-quarter of joins (26\%) are conditioned on one-to-one relationships, and 10\% of joins are conditioned on many-to-many
relationships.

\vspace{1mm}
\noindent\emph{\textbf{Self join.}}
28\% of queries include at least one self join, defined as a join in which the same database table appears in both joined relations. The remaining queries (72\%) contain no self joins.

\vspace{1mm}
\noindent\emph{\textbf{Join type.}}
Inner join is the most common join type (69\%), followed by left join (29\%) and cross join (1\%). The remaining types (right join and full join) together account for less than 1\%.

\question{\questionfour}

\rmspc
\begin{figure}[H]
  \begin{tabular}{c l}
    \includegraphics[width=0.14\textwidth]{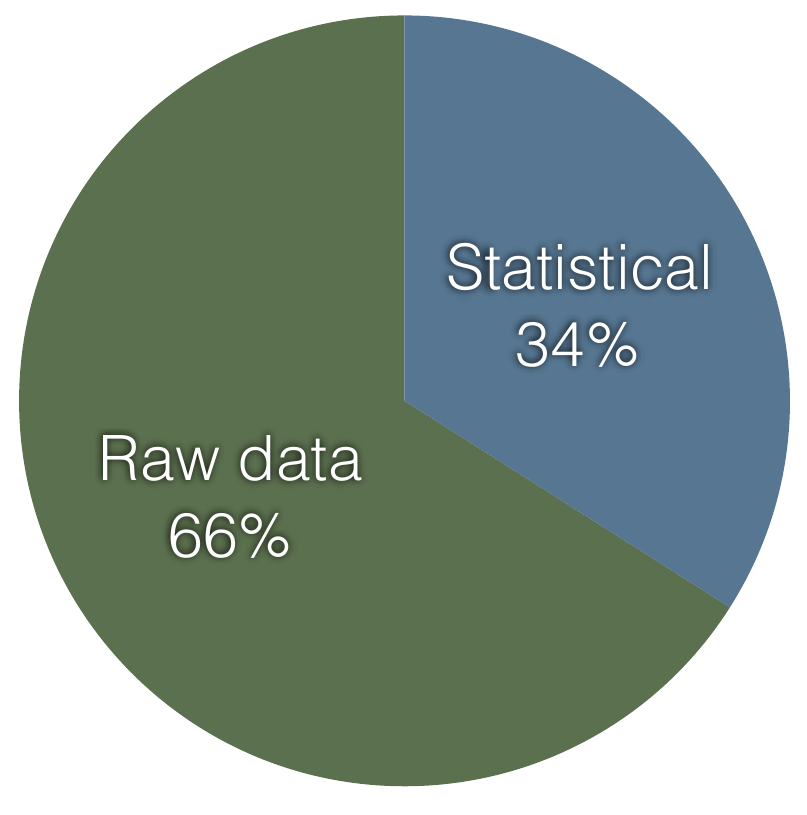}
    &
     \begin{minipage}{.29\textwidth}
        \vspace*{-65pt}
\results Approximately one-third of queries are statistical,
meaning they return only aggregations (count, average, etc.).
The remaining queries return non-aggregated results (i.e., raw data)
in at least one output column.
     \end{minipage}
  \end{tabular}
\end{figure}
\rmspc \rmspc

\question{\questionfive}

\rmspc
\begin{figure}[H]
  \begin{tabular}{c l}
    \includegraphics[width=0.14\textwidth]{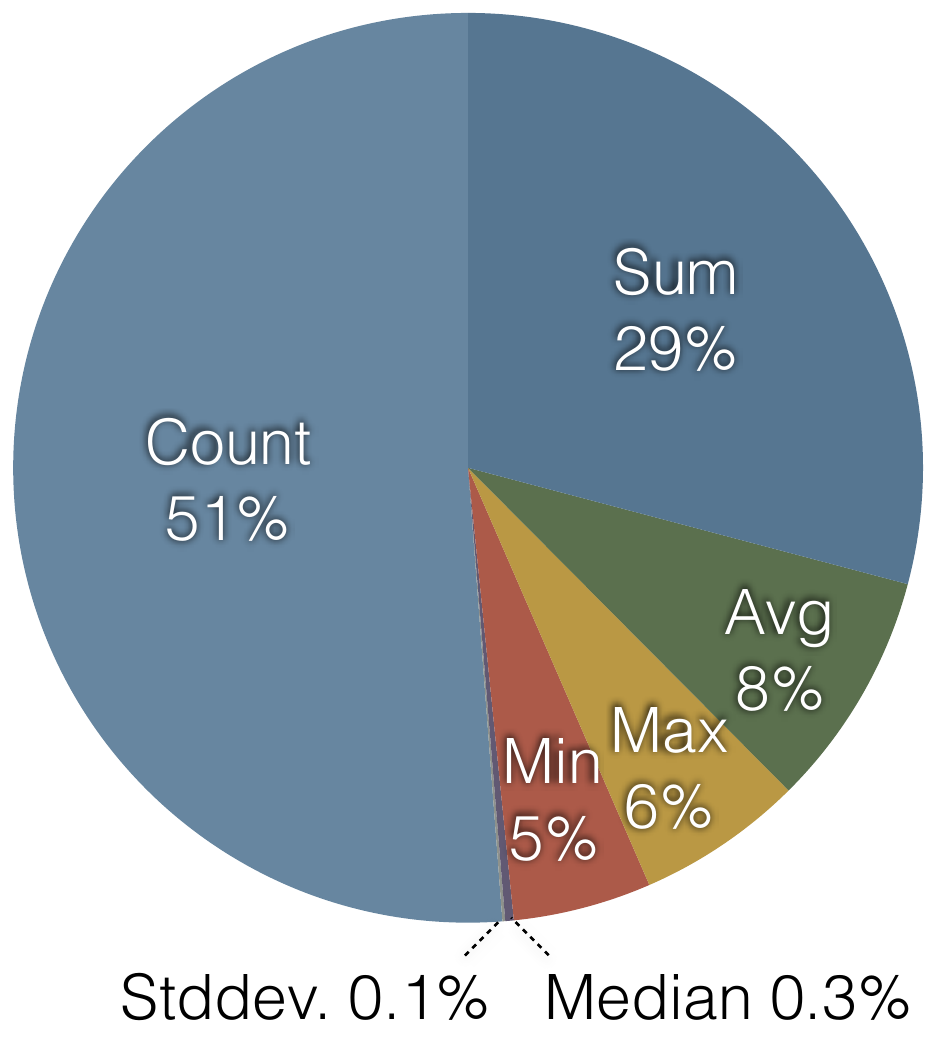}
    &
      \begin{minipage}{.29\textwidth}
        \vspace*{-65pt}
\results \texttt{Count} is the most common
aggregation function (51\%), followed by \texttt{Sum} (29\%),
\texttt{Avg} (8\%), \texttt{Max} (6\%) and \texttt{Min} (5\%).
The remaining functions account for fewer than 1\% of all
aggregation functions.
      \end{minipage}
  \end{tabular}
\end{figure}
\rmspc \rmspc

\question{\questionsixa}

\rmspc
 \begin{figure}[H]
  \centering
\includegraphics[width=.40\textwidth]{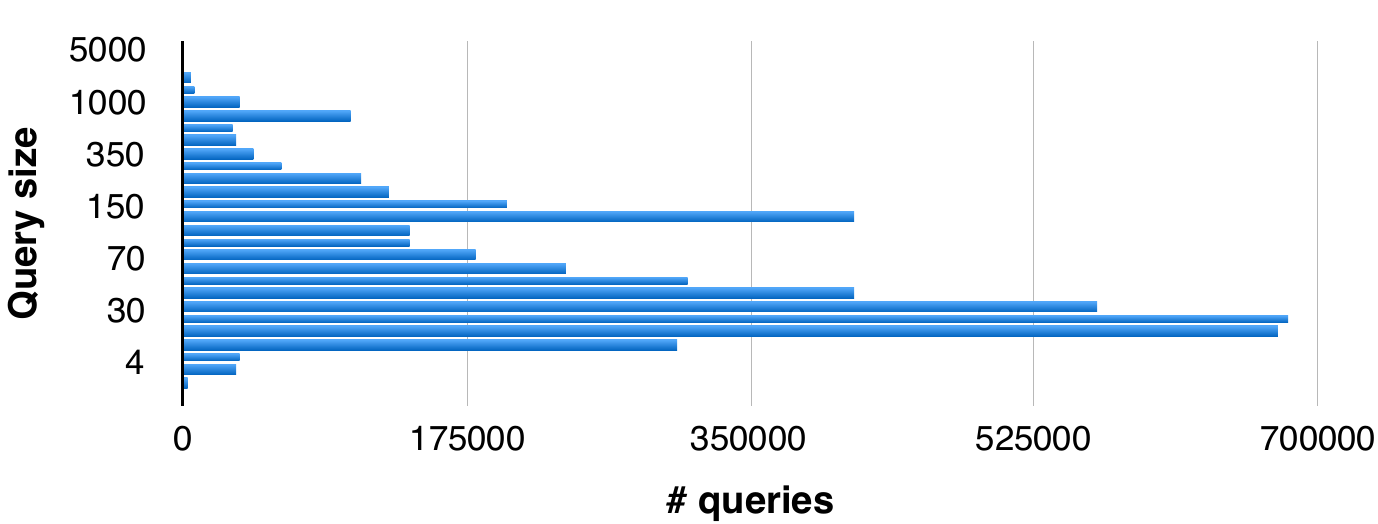}  
\end{figure}
\rmspc
\results The majority of queries are fewer than 100 clauses but a
significant number of queries are much larger, with some queries
containing as many as thousands of clauses. 

\question{\questionsixb}
\rmspc
\begin{figure}[H]
  \centering
\includegraphics[width=.40\textwidth]{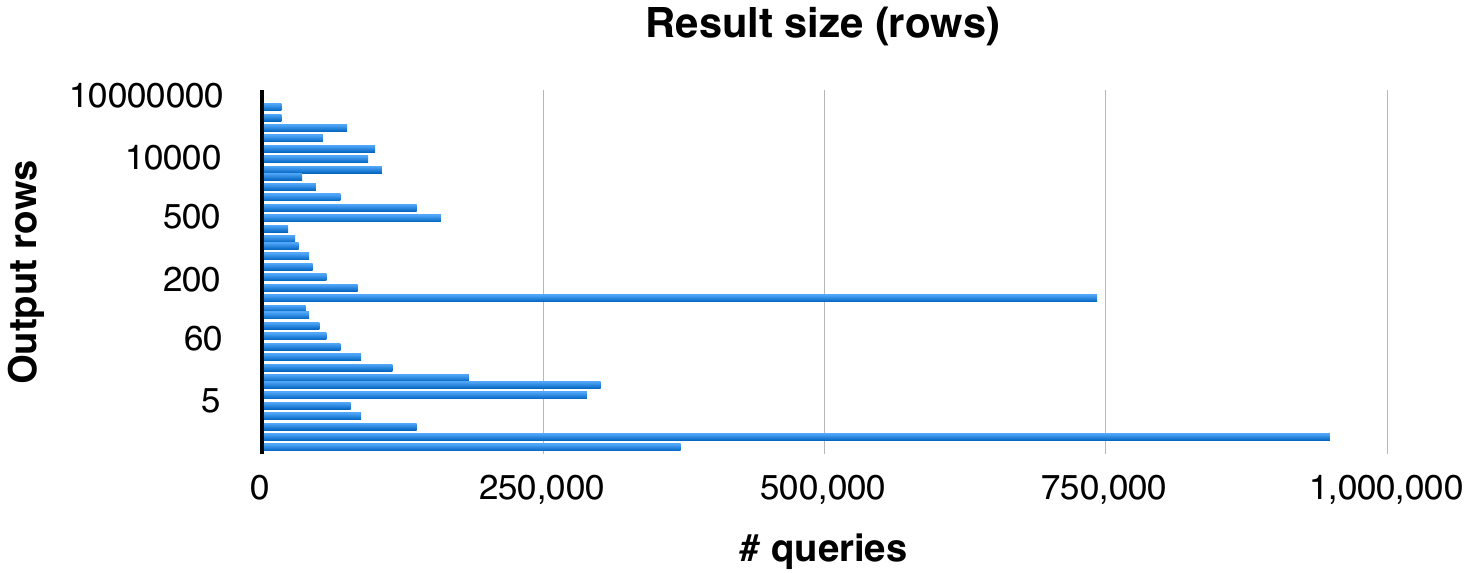}  
\end{figure}
\vspace{-4mm}
\begin{figure}[H]
  \centering
\includegraphics[width=.40\textwidth]{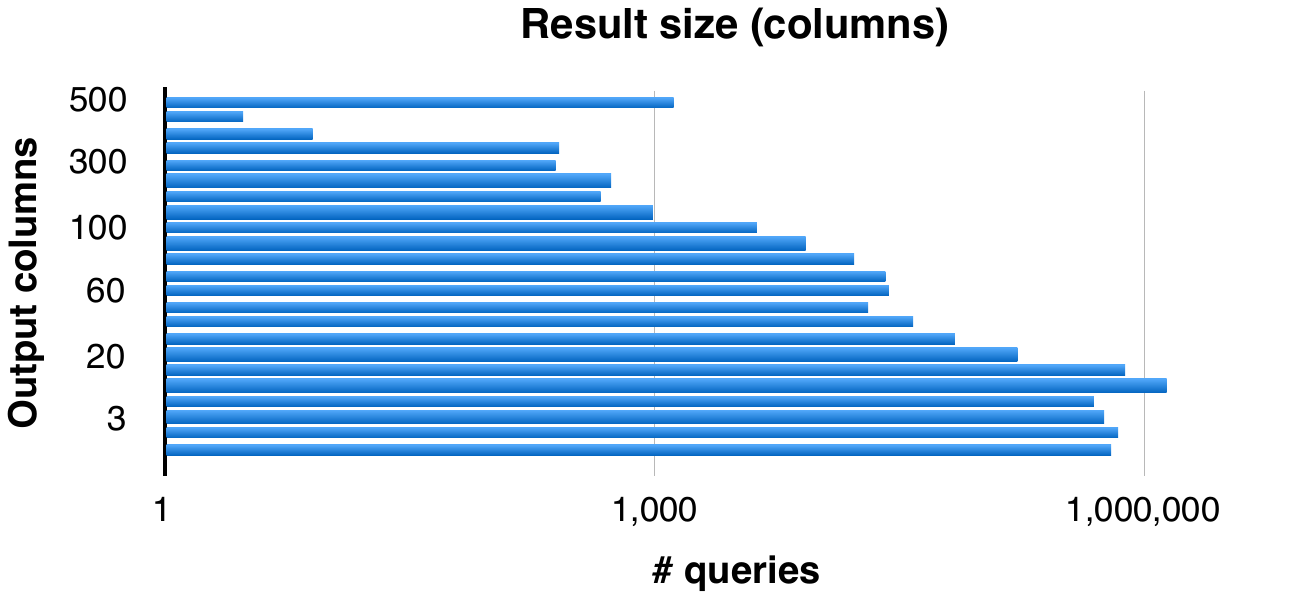}  
\end{figure}
\rmspc
\results 
The output sizes of queries varies dramatically with respect
to both rows and columns, and queries commonly return \emph{hundreds}
of columns and \emph{hundreds of thousands} of rows.

\fi

\def\xmark{}
\def\noteone{\textsuperscript{\ding{81}}}
\def\notetwo{\textsuperscript{\ding{81}}}

\begin{table*}
\small
\centering
\caption{Comparison of general-purpose
  differential privacy mechanisms with support for join.}
\label{tbl:comparison-of-mechanisms}
  \rowcolors{2}{gray!15}{white}
  \def\arraystretch{1.3}
  \begin{tabular}{l c c c c }
    & \vspace*{-4mm}$\overbrace{}^\textbf{\scriptsize Requirement 1}$ & \multicolumn{3}{c}{$\overbrace{\phantom{Equijoin Equijoin Equijoin Equijoin Equi}}^\textbf{\scriptsize Requirement 2}$} \\
    \rowcolor{white} & Database & One-to-one & One-to-many & Many-to-many \\
    \rowcolor{white} & compatibility & equijoin & equijoin & equijoin   \\
    \hline
    PINQ~\cite{mcsherry2009privacy}  & \xmark & \checkmark & \xmark & \xmark  \\
    wPINQ~\cite{proserpio2014calibrating} & \xmark & \checkmark & \checkmark & \checkmark  \\
    Restricted sensitivity~\cite{Blocki:2013:DPD:2422436.2422449} & \xmark & \checkmark & \checkmark & \xmark  \\
    DJoin~\cite{narayan2012djoin} & \xmark & \checkmark & \xmark & \xmark   \\
    \Staticsensitivity (this work) & \checkmark & \checkmark & \checkmark & \checkmark  \\
  \end{tabular}

\end{table*}

\ifvldb
\else
\subsection{Summary of Results}
\fi
Our study reveals that real-world queries are executed on
many different database engines---in our dataset there
are over 6. We believe this is typical; a variety of
databases are commonly used within a company to match specific
performance and scalability requirements. 
A practical mechanism for differential privacy will therefore allow the use of any of these
existing databases, requiring neither a specific database distribution nor
a custom execution engine in lieu of a standard database.\todo{To the best of our knowledge, no existing differential privacy mechanism satisfies this requirement.}

The study shows that 62.1\% of all queries use SQL \texttt{Join}, and specifically 
equijoins which are by far the most common. Additionally, a majority of queries use multiple joins, more than 
one-quarter use self joins, and joins are conditioned on one-to-one, one-to-many, and many-to-many relationships. These results suggest that a practical differential privacy approach must at a minimum provide robust 
support for equijoins, including the full spectrum of join relationships and an
arbitrary number of nested joins.

One-third (34\%) of all queries return aggregate
statistics. Differential privacy is principally designed for such queries,
and in the remainder of this paper we focus on these queries. Enforcing differential privacy for raw data queries is beyond the scope of this work, as differential privacy is
generally not intended to address this problem.

For statistical queries, \texttt{Count} is by far the most common
aggregation. This validates the focus on counting and histogram
queries by the majority of previous general-purpose differential privacy
mechanisms~\cite{mcsherry2009privacy,proserpio2014calibrating,Blocki:2013:DPD:2422436.2422449,narayan2012djoin}. Our work 
similarly focuses on this class of queries. In Section~\ref{sec:other-aggregation-functions} we discuss possible extensions to support other aggregation functions.

\vspace{1mm}
\paragraph{Requirements}
\label{sec:study-requirements}
We summarize our requirements for practical differential privacy of real-world SQL queries:

\begin{itemize}[topsep=0.5mm,leftmargin=4mm]
\itemsep0.5mm

\item \textbf{Requirement 1: Compatibility with existing databases.}
A practical differential privacy approach must support heterogeneous
database environments by not requiring a specific database distribution or
replacement of the database with a custom runtime.

\item \textbf{Requirement 2: Robust support for equijoin.}
A practical differential privacy approach must provide robust support for equijoin,
including both self joins and non-self joins, all join relationship types,
and queries with an arbitrary number of nested joins.

\end{itemize}

\vspace{1mm}
Our study shows that a differential privacy system satisfying these
requirements is likely to have broad impact, supporting a majority
of real-world statistical queries.

\subsection{Existing Differential Privacy Mechanisms}
\label{sec:exist-diff-priv}
\todo{remind that other work only supports count}

Several existing general-purpose differential privacy mechanisms support queries
with join. Table~\ref{tbl:comparison-of-mechanisms} summarizes these mechanisms and their 
supported features in comparison to our proposed mechanism (last row). 
In Section~\ref{sec:related-work} we discuss additional mechanisms which are not considered here
either because they are not general-purpose or they cannot support joins.

\paragraph{PINQ}
Privacy Integrated Queries (PINQ)~\cite{mcsherry2009privacy} is a
mechanism that provides differential privacy for counting queries
written in an augmented SQL dialect.  PINQ supports a restricted join
operator that groups together results with the same key. For
one-to-one joins, this operator is equivalent to the standard
semantics. For one-to-many and many-to-many joins, on the other hand,
a PINQ query can count the number of unique join
\emph{keys} but not the number of joined \emph{results}. Additionally,
PINQ introduces new operators that do not exist in standard SQL, so the approach
is not compatible with standard databases.

\paragraph{wPINQ}
Weighted PINQ (wPINQ)~\cite{proserpio2014calibrating} extends PINQ
with support for general equijoins and works by assigning a weight to
each row in the database, then scaling down the weights of rows in a
join to ensure an overall sensitivity of 1. In wPINQ, the result of a
counting query is the sum of weights for records being counted plus
noise drawn from the Laplace distribution. This approach allows wPINQ
to support all three types of joins. However, wPINQ does not satisfy our
database compatibility requirement. The system described by Proserpio et al.~\cite{proserpio2014calibrating} uses a custom runtime; applying wPINQ in an existing database would require
modifying the database to propagate weights during execution.

\paragraph{Restricted sensitivity}
Restricted sensitivity~\cite{Blocki:2013:DPD:2422436.2422449} is
designed to bound the global sensitivity of counting queries with
joins, by using properties of an auxiliary data model. The approach
requires bounding the frequency of each join key globally (i.e. for
all possible future databases). This works well for one-to-one and
one-to-many joins, because the unique key on the ``one'' side of the
join has a global bound. However, it cannot handle many-to-many joins,
because the frequency of keys on \emph{both} sides of the join may be
unbounded. Blocki et al.~\cite{Blocki:2013:DPD:2422436.2422449}
formalize the restricted sensitivity approach but do not describe how
it could be used in a system compatible with existing databases, and
no implementation is available.

\paragraph{DJoin}
DJoin~\cite{narayan2012djoin} is a mechanism designed for
differentially private queries over datasets distributed over multiple
parties. Due to the additional restrictions associated with this
setting, DJoin supports only one-to-one joins, because it rewrites
join queries as relational intersections. For example, consider the
following query:

\rmspcsm
\begin{lstlisting}
SELECT COUNT(*) FROM X JOIN Y ON X.A = Y.B
\end{lstlisting}
\rmspcsm

\noindent DJoin rewrites this query to the following (in relational
algebra), which is semantically equivalent to the original query
\emph{only} if the join is one-to-one: $|\pi_A(X) \cap \pi_B(Y)|$.

Additionally, the approach requires the use of special cryptographic
functions during query execution, so it is not compatible with existing
databases.

\vspace{1mm}

To address the limitations of existing mechanisms we propose
\textit{\staticsensitivity}, discussed next. Elastic
sensitivity is compatible with any existing database and supports
general equijoins with the full spectrum of join relationships.
This combination allows use of elastic sensitivity in real-world settings.

\section{\StaticSensitivity}
\label{sec:static-sensitivity}

\Staticsensitivity is a novel approach for calculating an upper
bound on a query's local sensitivity. After motivating the approach,
we provide background on necessary concepts in Section~\ref{sec:background},
formally define \staticsensitivity in Section~\ref{sec:defin-stat}, give an example in Section~\ref{sec:example}, prove its
correctness in Section~\ref{sec:corr-stat}, and discuss an optimization in Section~\ref{sec:optim-joins-with}.

\subsection{Motivation}

Many previous differential privacy mechanisms~\cite{mcsherry2009privacy,Blocki:2013:DPD:2422436.2422449} are
based on global sensitivity. These approaches do not generalize to
queries with joins; the global sensitivity of queries with general joins may
be \emph{unbounded} because ``a join has the ability to multiply input
records, so that a single input record can influence an arbitrarily
large number of output records.''~\cite{mcsherry2009privacy}

Techniques based on local sensitivity~\cite{nissim2007smooth,dwork2009differential} generally provide greater utility than global sensitivity-based approaches because they consider the \emph{actual} database.
Indeed, local sensitivity is finite for general queries with joins.
However, directly computing local sensitivity is computationally
infeasible, as it requires running the query on every possible
neighbor of the true database---in our environment this would require
running more than 1 billion queries for each original query. Previous work~\cite{nissim2007smooth} describes
efficient methods to calculate local sensitivity for a limited set of fixed queries
(e.g., the median of all values in the database) but these techniques do not apply to
general-purpose queries or queries with join.

These challenges are reflected in the design of previous mechanisms
listed in Table~\ref{tbl:comparison-of-mechanisms}.
PINQ and restricted sensitivity support only joins for
which global sensitivity can be bounded, and wPINQ scales weights
attached to the data during joins to ensure a global sensitivity of 1.
DJoin uses a measure of sensitivity unique to its distributed setting.
None of these techniques are based on local sensitivity.

Elastic sensitivity is the first tractable approach to leverage local sensitivity for
queries with general equijoins. 
The key insight of our approach is to model the impact of each join in the
query using precomputed metrics about the frequency of join keys in the true database.
This novel approach allows elastic sensitivity to compute a \emph{conservative} approximation
of local sensitivity without requiring any additional interactions with the database.
In Section~\ref{sec:corr-stat}, we prove elastic sensitivity is an upper bound on local sensitivity and can
therefore be used with any local sensitivity-based differential
privacy mechanism. In Section~\ref{sec:apply-smooth-sens}, we describe
how to use elastic sensitivity to enforce differential privacy.

\subsection{Background}
\label{sec:background}

We briefly summarize existing differential privacy concepts necessary for  
describing our approach. For a more thorough overview of differential
privacy, we refer the reader to Dwork and Roth's excellent
reference~\cite{dwork2014algorithmic}.

Differential privacy provides a formal guarantee of
\emph{indistinguishability}: a differentially private result does not
yield very much information about which of two neighboring databases
was used in calculating the result. 

  Formally, differential privacy considers a database modeled as a vector $x \in D^n$, in which $x_i$ represents the data contributed by user $i$. The \emph{distance} between two databases $x, y \in D^n$ is $d(x,y) = |\{i | x_i \not = y_i\}|$. Two databases $x, y$ are \emph{neighbors} if $d(x,y) = 1$. 

\begin{definition}[Differential privacy]
A randomized mechanism $\mathcal{K} : D^n \rightarrow
\mathbb{R}^d$ preserves $(\epsilon, \delta)$-differential privacy if
for any pair of databases $x, y \in D^n$ such that $d(x,y) = 1$, and
for all sets $S$ of possible outputs:

\[\mbox{Pr}[\mathcal{K}(x)\in S] \leq e^\epsilon
  \mbox{Pr}[\mathcal{K}(y) \in S] + \delta \]
\end{definition}

Intuitively,
the \emph{sensitivity} of a query corresponds to the amount its results can
change when the database changes.
One measure of sensitivity is \emph{global sensitivity}, which is the maximum difference in
the query's result on \emph{any} two neighboring databases.

\begin{definition}[Global Sensitivity] For $f:D^n \rightarrow
  \mathbb{R}^d$ and all $x,y \in D^n$, the global sensitivity of $f$
  is
  
\[GS_f = \max_{x,y:d(x,y)=1}||f(x)-f(y)||\]
\end{definition}

McSherry~\cite{mcsherry2009privacy} defines the notion of \emph{stable transformations} on a database, which we will use later. Intuitively, a transformation is stable if its privacy implications can be bounded.

\begin{definition}[Global Stability] A transformation $T: D^n \rightarrow D^n$ is \emph{$c$-stable} if for $x, y \in D^n$ such that $d(x, y) = 1$, $d(T(x), T(y)) \leq c$.
\end{definition}

Another definition of sensitivity is \emph{local
  sensitivity}~\cite{dwork2009differential,nissim2007smooth}, which 
is the maximum difference between the query's results on the
\emph{true database} and any neighbor of it:

\begin{definition}[Local Sensitivity] For $f:D^n \rightarrow
  \mathbb{R}^d$ and $x \in D^n$, the local sensitivity of $f$ at $x$
  is
  
\[LS_f(x) = \max_{y:d(x,y)=1}||f(x)-f(y)||\]
\end{definition}

Local sensitivity is often much
lower than global sensitivity since it is a property of the single true
database rather than the set of all possible databases.

We extend the notion of stability to the case of local sensitivity by fixing $x$ to be the true database.

\begin{definition}[Local Stability] A transformation $T: D^n \rightarrow D^n$ is \emph{locally $c$-stable} for true database $x$ if for $y \in D^n$ such that $d(x, y) = 1$, $d(T(x), T(y)) \leq c$.
  \label{def:local_stability}
\end{definition}

\paragraph{Differential privacy for multi-table databases}
In this paper we consider \emph{bounded} differential privacy~\cite{kifer2011no}, in which $x$ can be obtained from its neighbor $y$ by changing (but not adding or removing) a single tuple. Our setting involves a database represented as a multiset of tuples, and we wish to protect the presence or absence of a single tuple. If tuples are drawn from the domain $D$ and the database contains $n$ tuples, the setting can be represented as a vector $x \in D^n$, in which $x_i = v$ if row $i$ in the database contains the tuple $v$.

For queries without joins, a database $x \in D^n$ is considered as a single table. However, our setting considers database with multiple tables and queries with joins. We map this setting into the traditional definition of differential privacy by considering $m$ tables $t_1, ..., t_m$ as disjoint subsets of a single database $x \in D^n$, so that $\bigcup^m_{i = 1} t_i = x$.

With this mapping, differential privacy offers the same protection as in the single-table case: it protects the presence or absence of any single tuple in the database. When a single user contributes more than one protected tuple, however, protecting individual tuples may not be sufficient to provide privacy. Note that this caveat applies \emph{equally} to the single- and multi-table cases---it is not a unique problem of multi-table differential privacy.

We maintain the same definition of neighboring databases as the single-table case. Neighbors of $x \in D^n$ can be obtained by selecting a table $t_i \in x$ and changing a single tuple, equivalent to changing a single tuple in a single-table database.

\paragraph{Smoothing functions}
Because local sensitivity is based on the true database, it must
be used carefully to avoid leaking information about the data.
Prior work~\cite{nissim2007smooth,dwork2009differential}
describes techniques for using local sensitivity to enforce differential privacy.
Henceforth we use the term \emph{smoothing functions} to refer to
these techniques. Smoothing functions are independent of the method used
to compute local sensitivity, but generally require that local
sensitivity can be computed an arbitrary distance $k$ from the true
database (i.e. when at most $k$ entries are changed).

\begin{definition}[Local Sensitivity at Distance]
  The local sensitivity of $f$ at distance $k$ from database $x$ is:
  
  \[A_f^{(k)}(x) =\max_{y \in D^n:d(x,y) \leq k}LS_f(y)\]  
  \label{def:a}
\end{definition}

\rmspc
\subsection{Definition of \StaticSensitivity}
\label{sec:defin-stat}

We define the \emph{\staticsensitivity} of a query recursively on the
query's structure. To allow the use of smoothing functions,
our definition describes how to calculate elastic sensitivity at
arbitrary distance $k$ from the true database (under this definition,
the local sensitivity of the query is defined at $k=0$).

Figure~\ref{fig:staticsensitivity-defn} contains the
complete definition, which is in four parts:
(a) \emph{Core relational algebra},
(b) \emph{Definition of \Staticsensitivity},
(c) \emph{Max frequency at distance $k$}, and
(d) \emph{Ancestors of a relation}. We describe each part next.

\paragraph{Core relational algebra}
We present the formal definition of elastic sensitivity in terms of a
subset of the standard relational algebra, defined in
Figure~\ref{fig:staticsensitivity-defn}(a). This subset includes
selection ($\sigma$), projection ($\pi$), join ($\bowtie$), counting
($\mbox{\emph{Count}}$), and counting with grouping
($\mbox{\emph{Count}}_{G_1..G_n}$). It admits arbitrary equijoins,
including self joins, and all join relationships (one-to-one, one-to-many, and
many-to-many).

To simplify the presentation our notation assumes the query
performs a count as the \emph{outermost} operation, however
the approach naturally extends to aggregations nested anywhere in the query
as long as the query does not perform arithmetic or other modifications
to aggregation result. For example, the following
query counts the total number of trips and projects the ``count''
attribute:

\[ \pi_{\mbox{count}}\mbox{\emph{Count}}(\mbox{trips}) \]

\noindent Our approach can support this query by treating the inner relation
as the query root.

\paragraph{\Staticsensitivity}
Figure~\ref{fig:staticsensitivity-defn}(b) contains the recursive
definition of \staticsensitivity at distance $k$. We denote the \staticsensitivity of
query $q$ at distance $k$ from the true database $x$ as
$\sfun^{(k)}(q, x)$. The $\sfun$ function is defined in terms of the
elastic \emph{stability} of relational transformations (denoted
$\sfun_R$).

$\sfun^{(k)}_R(r, x)$ bounds the local stability (Definition~\ref{def:local_stability}) of relation $r$ at
distance $k$ from the true database $x$. $\sfun^{(k)}_R(r, x)$ is defined in terms of $\mf_k(a, r,
x)$, the maximum frequency of attribute $a$ in relation $r$ at
distance $k$ from database $x$.

\paragraph{Max frequency at distance $k$}
The \emph{maximum frequency} metric is used to bound the sensitivity
of joins. We define the maximum frequency $\mf(a, r,x)$ as the
frequency of the most frequent value of attribute $a$ in relation $r$
in the database instance $x$. In Section~\ref{sec:flex} we describe how the
values of $\mf$ can be obtained from the database.

To bound the local sensitivity of a query at distance $k$ from the
true database, we must also bound the max frequency of each join key
at distance $k$ from the true database. For attribute $a$ of relation
$r$ in the true database $x$, we denote this value $\mf_k(a, r, x)$,
and define it (in terms of $\mf$) in Figure~\ref{fig:staticsensitivity-defn}(c).

\paragraph{Ancestors of a relation}
The definition in Figure~\ref{fig:staticsensitivity-defn}(d) is a
formalization to identify self joins.
Self joins have a much greater effect on sensitivity than joins of
non-overlapping relations. In a self join, adding or removing one row of the
underlying database may cause changes in \emph{both} joined relations,
rather than just one or the other. The join case of \staticsensitivity
is therefore defined in two cases: one for self joins, and one for
joins of non-overlapping relations. To distinguish the two cases, we
use $\mathcal{A}(r)$ (defined in
Figure~\ref{fig:staticsensitivity-defn}(d)), which denotes the set of
tables possibly contributing rows to $r$.  A join of two relations
$r_1$ and $r_2$ is a self join when $r_1$ and $r_2$ \emph{overlap},
which occurs when some table $t$ in the underlying database
contributes rows to both $r_1$ and $r_2$. Rows $r_1$ and $r_2$ are
non-overlapping when $|\mathcal{A}(r_1) \cap \mathcal{A}(r_2)| = 0$.

\begin{figure*}
  \scriptsize
  \centering
  \hspace*{-.7cm}
  \begin{tabular}{l l l}
    \begin{minipage}[t]{.2\linewidth}
      ~~\textbf{Core relational algebra:}
      
      $$
      \begin{array}{lcl}
      \multicolumn{3}{l}{\cellcolor{gray!15}\mbox{Attribute names}}\\
        a & &\\
        \\
        \multicolumn{3}{l}{\cellcolor{gray!15}\mbox{Value constants}}\\
        v & &\\

        \\
        \multicolumn{3}{l}{\cellcolor{gray!15}\mbox{Relational transformations}}\\ 
        R & ::= & t \;|\; R_1 \underset{x = y}{\bowtie} R_2 \\
          & | & \Pi_{a_1, ..., a_n} R \;|\; \sigma_\varphi R \vspace*{1mm}\\
          & | & \mbox{\emph{Count}}(R) \\
        \\
        \multicolumn{3}{l}{\cellcolor{gray!15}\mbox{Selection predicates}}\\
        \varphi & ::= & a_1 \theta a_2 \;|\; a \theta v\\
        \theta & ::= & < \;|\; \leq \;|\; = \\
        & | & \neq \;|\; \geq \;|\; >\\
        \\
        \multicolumn{3}{l}{\cellcolor{gray!15}\mbox{Counting queries}}\\ 
        Q & ::= & \mbox{\emph{Count}}(R) \\
         & | & \underset{G_1..G_n}{\mbox{\emph{Count}}}(R) \\
      \end{array}
      $$
      \centering
      \vfill
    \end{minipage}
    &

  \begin{minipage}[t]{.40\linewidth}
    ~~\textbf{Definition of elastic stability:}
    
  $$
  \begin{array}{l c l}

    \cellcolor{gray!15}\sfun^{(k)}_R &\cellcolor{gray!15}::& \cellcolor{gray!15}R \rightarrow D^n\!\! \rightarrow \mbox{\emph{elastic stability}} \\
    \sfun^{(k)}_R(t, x) &=& 1 \\
    \sfun^{(k)}_R(r_1 \underset{a = b}{\bowtie} r_2, x) &=& \\
    \multicolumn{3}{l}{
    \hfill\begin{cases}
      \max(\mf_k(a, r_1, x)\sfun^{(k)}_R(r_2, x), \\ \phantom{\max(}\mf_k(b, r_2, x)\sfun^{(k)}_R(r_1, x)) & | \mathcal{A}(r_1) \cap \mathcal{A}(r_2) | = 0  \\
      \\
            \mf_k(a, r_1, x)\sfun^{(k)}_R(r_2, x) + \\ \;\;\mf_k(b, r_2, x)\sfun^{(k)}_R(r_1, x) + \\ \;\;\sfun^{(k)}_R(r_1, x)\sfun^{(k)}_R(r_2, x) & | \mathcal{A}(r_1) \cap \mathcal{A}(r_2) | > 0 \\

    \end{cases}}\vspace*{1mm}\\

    \\
    \sfun^{(k)}_R(\Pi_{a_1, ..., a_n} r, x) &=& \sfun^{(k)}_R(r, x)\\
    \sfun^{(k)}_R(\sigma_\varphi r, x) &=& \sfun^{(k)}_R(r, x)\\
    \ifsubqueries
    \sfun^{(k)}_R(\mbox{\emph{Count}}(r)) &=& 1 \\[2mm]
    \fi

\multicolumn{2}{l}{\hspace{-2mm}\textbf{Definition of \staticsensitivity:}}\vspace{0.5mm}
    \\
    \cellcolor{gray!15}\sfun^{(k)} &\cellcolor{gray!15}::& \cellcolor{gray!15}Q \rightarrow D^n\!\! \rightarrow \mbox{\emph{\staticsensitivity}} \\
    \sfun^{(k)}(\mbox{\emph{Count}}(r), x) &=& \sfun^{(k)}_R(r, x)\\
    \sfun^{(k)}(\underset{G_1..G_n}{\mbox{\emph{Count}}}(r), x) &=& 2 \sfun^{(k)}_R(r, x)\\

  \end{array}
  $$
\end{minipage}

&
    \begin{minipage}[t]{.31\linewidth}
      \textbf{Maximum frequency at distance $k$:}
      
    $$
  \begin{array}{l c l}
    \cellcolor{gray!15}\mf_k &\cellcolor{gray!15}::& \cellcolor{gray!15}a \rightarrow R \rightarrow D^n\!\! \rightarrow \mathbb{N}\\
    \mf_k(a, t, x) &=& \mf(a, t, x) + k\\

    \todo{this is a revision}\\
    \mf_k(a_1, r_1 \underset{a_2 = a_3}{\bowtie} r_2,x) &=&  \\

    \multicolumn{3}{l}{
    \hfill\begin{cases}
      \mf_k(a_1, r_1, x) \mf_k(a_3, r_2, x) & a_1 \in r_1\\
      \mf_k(a_1, r_2, x) \mf_k(a_2, r_1, x) & a_1 \in r_2\\
    \end{cases}}\vspace*{1mm}\\
    \\

    \mf_k(a, \Pi_{a_1, ..., a_n} r, x) &=& \mf_k(a, r, x)\\
    \mf_k(a, \sigma_\varphi r, x) &=& \mf_k(a, r, x)\\
    \mf_k(a, \mbox{\emph{Count}}(r), x) &=& \bot\\

  \end{array}
  $$

  \begin{center}
    (c)
  \end{center}
  
  \textbf{Ancestors of a relation:}
  
    $$
  \begin{array}{l c l}
    \cellcolor{gray!15}\mathcal{A} &\cellcolor{gray!15}::& \cellcolor{gray!15}R \rightarrow \{R\}\\
    \mathcal{A}(t) &=& \{t\}\\
    \mathcal{A}(r_1 \underset{a = b}{\bowtie} r_2) &=& \mathcal{A}(r_1) \cup \mathcal{A}(r_2)\\
    \mathcal{A}(\Pi_{a_1, ..., a_n} r) &=& \mathcal{A}(r)\\
    \mathcal{A}(\sigma_\varphi r) &=& \mathcal{A}(r)\\
  \end{array}
  $$

\end{minipage}
                   \vspace*{-1mm}
\end{tabular}
\begin{tabular}{l l l}
  \small
\hspace{-2cm}(a)\hspace*{5cm}&(b)\hspace*{6.5cm}&(d)
\end{tabular}
\caption{(a) syntax of core relational algebra; (b) definition of
  elastic stability and \staticsensitivity at distance $k$; (c) definition of maximum
  frequency at distance $k$; (d) definition of ancestors of a
  relation.}\label{fig:staticsensitivity-defn}
\label{fig:staticsensitivity-defn}
\end{figure*}

\paragraph{Join conditions}
For simplicity our notation refers only to the case where a join
contains a single equality predicate. However, the approach naturally extends to join conditions containing \emph{any} predicate that can be decomposed into a conjunction of an equijoin term and any other terms.
Consider for example the following query:

\begin{lstlisting}[language=sql]
SELECT COUNT(*) FROM a
JOIN b ON a.id = b.id AND a.size > b.size
\end{lstlisting}
\vspace{-1mm}

Calculation of elastic sensitivity for this query requires only the equijoin
term ($a.id = b.id$) and therefore follows directly from our definition.
Note that in a conjunction, each predicate adds additional constraints
that may decrease (but never increase) the true local stability of the join,
hence our definition correctly computes an upper bound on the stability.

\paragraph{Unsupported queries} We discuss several cases of queries
that are not supported by our definition in
Section~\ref{sec:unsupported-queries}.

\subsection{Example: Counting Triangles}
\label{sec:example}

We now consider step-by-step calculation of elastic sensitivity for an example query. We select the problem of counting triangles in a directed graph, described by Prosperio et al. in their evaluation of WPINQ~\cite{proserpio2014calibrating}. This example contains multiple self-joins, which demonstrate the most complex recursive cases of Figure~\ref{fig:staticsensitivity-defn}.

Following Prosperio et al. we select privacy budget $\epsilon = 0.7$
and consider the~\texttt{ca-HepTh}~\cite{hepth} dataset, which has maximum frequency metric of 65.

In SQL, the query is expressed as:
  
\begin{lstlisting}
SELECT COUNT(*) FROM edges e1
JOIN edges e2 ON e1.dest = e2.source AND 
                 e1.source < e2.source
JOIN edges e3 ON e2.dest = e3.source AND 
                 e3.dest = e1.source AND 
                 e2.source < e3.source
\end{lstlisting}

Consider the first join ($e_1 \bowtie e_2$), which joins the \texttt{edges} table with itself. By definition of $\sfun_R^{(k)}$ (self join case) the elastic stability of this relation is:

{\scriptsize
\begin{align*}
    \mf_k(\mbox{\emph{dest}}, \mbox{\texttt{edges}}, x)\sfun^{(k)}_R(\mbox{\texttt{edges}}, x) + \\ \;\;\mf_k(\mbox{\emph{source}}, \mbox{\texttt{edges}}, x)\sfun^{(k)}_R(\mbox{\texttt{edges}}, x) + \\ \;\;\sfun^{(k)}_R(\mbox{\texttt{edges}}, x)\sfun^{(k)}_R(\mbox{\texttt{edges}}, x)
\end{align*}
}
Furthermore, since \texttt{edges} is a table, {\scriptsize $\sfun_R^{(k)}(\mbox{\texttt{edges}}) = 1$.}

\hspace{-2.5mm}We then have:

{\scriptsize
\begin{align*}
  &\mf_k(\mbox{\emph{dest}}, \mbox{\texttt{edges}}, x) = \mf(\mbox{\emph{dest}}, \mbox{\texttt{edges}}, x) + k\\ &\mf_k(\mbox{\emph{source}}, \mbox{\texttt{edges}}, x) = \mf(\mbox{\emph{source}}, \mbox{\texttt{edges}}, x) + k
\end{align*}
}
Substituting the max frequency metric (65), the elastic stability of this relation is
  $(65+k) + (65+k) + 1 = 131 + 2k$.

Now consider the second join, which joins $e_3$ (an alias for the \texttt{edges} table) with the previous joined relation ($e_1 \bowtie e_2$). Following the same process and substituting values, the elastic stability of this relation is:

{\scriptsize
\begin{align*}
   &\mf_k(\mbox{\emph{dest}}, \mbox{\texttt{edges}}, x)\sfun^{(k)}_R(e_1 \bowtie e_2, x) + \\
  &\;\;\;\;\mf_k(\mbox{\emph{source}}, \mbox{\texttt{edges}}, x)\sfun^{(k)}_R(\mbox{\texttt{edges}}, x) + \\
  &\;\;\;\;\sfun^{(k)}_R(e_1 \bowtie e_2, x)\sfun^{(k)}_R(\mbox{\texttt{edges}}, x)\\
  = &(65+k)(131 + 2k) + 
      (65+k) + (131 + 2k)\\
  = &2k^2 + 199k + 8711
\end{align*}
}

\vspace{-4mm}
This expression describes the elastic stability at distance $k$ of relation $(e_1 \bowtie e_2)\bowtie e_3$. Per the definition of $\sfun^{(k)}$ the elastic sensitivity of a counting query is equal to the elastic stability of the relation being counted, therefore this expression defines the elastic sensitivity of the full original query.

As we will discuss in Section~\ref{sec:proof-correctness}, elastic sensitivity must be
smoothed using smooth sensitivity~\cite{nissim2007smooth} before it can be used with the Laplace mechanism.
In short, this process requires computing the maximum value of 
elastic sensitivity at $k$ multiplied by an exponentially decaying function in $k$:

{\scriptsize
\begin{align*}
   &S = \max_{k=0,1,...,n} e^{-\beta k} \sfun^{(k)} \\
  &\;\;\;= \max_{k=0,1,...,n} e^{-\beta k} (2k^2 + 199k + 8711)
\end{align*}
}
where {\scriptsize $\beta = \frac{\epsilon}{2 \ln(2/\delta)}$} and {\scriptsize $\delta = 10^{-8}$}.

The maximum value is $S = 8896.95$, which occurs at distance $k = 19$.
Therefore, to enforce differential privacy we add Laplace noise scaled to
$\frac{2S}{\epsilon} = \frac{17793.9}{0.7}$,
per Definition~\ref{def:flex} (see Section~\ref{sec:proof-correctness}).

\subsection{Elastic Sensitivity is an Upper Bound on Local Sensitivity}
\label{sec:corr-stat}

In this section, we prove that \staticsensitivity is an upper bound on
the local sensitivity of a query. This fundamental result affirms the
soundness of using elastic sensitivity in any local sensitivity-based
differential privacy mechanism. First, we prove two important
lemmas: one showing the correctness of the max frequency at distance
$k$, and the other showing the correctness of elastic stability.

\begin{lemma}
  For database $x$, at distance $k$, $r$ has at most $\mf_k(a, r, x)$
  occurrences of the most popular join key in attribute $a$: 
 \vspace{-1mm}
{\small \[\mf_k(a, r, x) \geq \max_{y:d(x,y) \leq k} \mf(a, r, y) \] }
\label{lma:mf}
\end{lemma}
\vspace{-6mm}
\begin{proof} By induction on the structure of $r$.
  
  \case{Case $t$}
  To obtain the largest possible number of occurrences of the most popular join key in a table $t$ at distance $k$, we modify $k$ rows to contain the most popular join key. Thus,
  $\max_{y:d(x,y) \leq k} \mf(a, r, y) = \mf(a, r, x) + k$.

  \case{Case $r_1 \underset{a_2 = a_3}{\bowtie} r_2$}
 We need to show that:

\smallskip
\smallskip
\smallskip
{\scriptsize
  \begin{equation}
    \mf_k(a_1, r_1 \underset{a_2 = a_3}{\bowtie} r_2, x) \geq \max_{y:d(x,y) \leq k} \mf(a_1, r_1 \underset{a_2 = a_3}{\bowtie} r_2, y)
    \label{mf_proof_to_prove}
  \end{equation}
}

\noindent Consider the case when $a_1 \in r_1$ (the proof for case $a_1 \in r_2$ is symmetric). The worst-case sensitivity occurs when each tuple in $r_1$ with the most popular value for $a_1$ also contains attribute value $a_2$ matching the most popular value of attribute $a_3$ in $r_2$. So we can rewrite equation~\ref{mf_proof_to_prove}:

\smallskip
{\scriptsize
  \begin{equation}
    \mf_k(a_1, r_1 \underset{a_2 = a_3}{\bowtie} r_2, x) \geq \max_{y:d(x,y) \leq k} \mf(a_1, r_1, y) \mf(a_3, r_2, y)
    \label{mf_proof_rewritten}
  \end{equation}
}

  \noindent We then rewrite the left-hand side, based on the definition of $\mf_k$ and the inductive hypothesis. Each step may make the left-hand side smaller, but never larger, preserving the original inequality:

{\scriptsize
  \begin{align*}
    & \mf_k(a_1, r_1 \underset{a_2 = a_3}{\bowtie} r_2, x) \\
    & = \mf_k(a_1, r_1, x) \mf_k(a_3, r_2, x) \\
    & \geq \max_{y:d(x,y) \leq k} \mf(a_1, r_1, y) \max_{y:d(x,y) \leq k} \mf(a_3, r_2, y) \\
    & \geq \max_{y:d(x,y) \leq k} \mf(a_1, r_1, y) \mf(a_3, r_2, y)
  \end{align*}
}

\vspace{-2mm}
  \noindent Which is equal to the right-hand side of equation~\ref{mf_proof_rewritten}.

  \case{Case $\Pi_{a_1, ..., a_n} r$}
  Projection does not change the number of rows, so the conclusion follows directly from the inductive hypothesis.

  \case{Case $\sigma_\varphi r$}
  Selection might filter out some rows, but does not modify attribute values. In the worst case, no rows are filtered out, so $\sigma_\varphi r$ has the same number of occurrences of the most popular join key as $r$. The conclusion thus follows directly from the inductive hypothesis.
\end{proof}

\begin{lemma}
$\sfun_R^{(k)}(r)$ is an upper bound on the local
  stability of relation expression $r$ at distance $k$ from database $x$:

  {\small \[ A_{\mbox{count}(r)}^{(k)}(x) \leq \sfun_R^{(k)}(r,x) \]  }
  \label{sr_lemma}
\end{lemma}
 \rmspc \rmspc
\begin{proof}
  By induction on the structure of $r$.

  \case{Case $t$} The stability of a table is 1, no matter its contents.

  \case{Case $r_1 \underset{a = b}{\bowtie} r_2$} We want to bound the number of changed rows in the joined relation. There are two
  cases, depending on whether or not the join is a self join.

  \case{Subcase 1: no self join} When the ancestors of $r_1$ and $r_2$
  are non-overlapping (i.e. $| \mathcal{A}(r_1) \cap \mathcal{A}(r_2)
  | = 0$), then the join is not a self join. This means that either
  $r_1$ may change or $r_2$ may change, \emph{but not both}. As a
  result, either $\sfun^{(k)}_R(r_1, x) = 0$ or $\sfun^{(k)}_R(r_2, x)
  = 0$. We therefore have two cases:
  
  \vspace{1mm}
  \begin{enumerate}[topsep=0.2mm,leftmargin=4mm]
  \item When $\sfun^{(k)}_R(r_1, x) = 0$, $r_2$ may contain at most
    $\sfun^{(k)}_R(r_2, x)$ changed rows, producing at most $\mf_k(a,
    r_1, x)\sfun^{(k)}_R(r_2, x)$ changed rows in the joined relation.

  \item In the symmetric case, when $\sfun^{(k)}_R(r_2, x) = 0$, the
    joined relation contains at most $\mf_k(b, r_2,
    x)\sfun^{(k)}_R(r_1, x)$ changed rows.
  \end{enumerate}

  \noindent We choose to modify the relation resulting in the largest
  number of changed rows, which is exactly the definition.

  \case{Subcase 2: self join} When the set of ancestor tables of $r_1$
  overlaps with the set of ancestor tables of $r_2$, i.e. $|
  \mathcal{A}(r_1) \cap \mathcal{A}(r_2) | > 0$, then changing a
  single row in the database could result in changed rows in
  \emph{both} $r_1$ and $r_2$.

  \vspace{2mm}
  \noindent In the self join case, there are three sources of changed rows:
  \begin{enumerate}[topsep=0.2mm,leftmargin=4mm]
  \item The join key of an original row from $r_1$ could match the join key of a changed row in $r_2$.
  \item The join key of an original row from $r_2$ could match the join key of a changed row in $r_1$.
  \item The join key of a changed row from $r_1$ could match the join key of a changed row in $r_2$.
  \end{enumerate}

  \vspace{2mm}
  \noindent Now consider how many changed rows could exist in each class.
    \begin{enumerate}[topsep=0.2mm,leftmargin=4mm]
  \item In class 1, $r_2$ could have at most
    $\sfun^{(k)}_R(r_2, x)$ changed rows (by the inductive
    hypothesis). In the worst case, each of these changed rows matches
    the \emph{most popular} join key in $r_1$, which occurs at most
    $\mf_k(a, r_1, x)$ times (by Lemma~\ref{lma:mf}), so class 1
     contains at most $\mf_k(a, r_1, x) \sfun^{(k)}_R(r_2,
    x)$ changed rows.
  \item Class 2 is the symmetric case of class 1, and thus contains at
    most $\mf_k(b, r_2, x) \sfun^{(k)}_R(r_1, x)$ changed rows.
  \item In class 3, we know that $r_1$ contains at most
    $\sfun^{(k)}_R(r_1,x)$ changed rows and $r_2$ contains at most
    $\sfun^{(k)}_R(r_2,x)$ changed rows. In the worst case, all of
    these changed rows contain the same join key, and so the joined
    relation contains $\sfun^{(k)}_R(r_1,x)\sfun^{(k)}_R(r_2,x)$
    changed rows.
  \end{enumerate}
  
  \noindent The total number of changed rows is therefore bounded by the sum of the bounds on the three classes:

{\scriptsize
  \begin{align*}
    \mf_k(a, r_1, x)\sfun^{(k)}_R(r_2, x) +
    &\mf_k(b, r_2, x)\sfun^{(k)}_R(r_1, x) +\\
    \sfun^{(k)}_R(r_1, x)&\sfun^{(k)}_R(r_2, x)
  \end{align*}
}
\vspace{-1mm}
  \noindent Which is exactly the definition.

  \case{Case $\Pi_{a_1, ..., a_n} r$}
  Projection does not change rows. The conclusion therefore follows from the inductive hypothesis.

  \case{Case $\sigma_\varphi r$}
  Selection does not change rows. The conclusion therefore follows from the inductive hypothesis.

  \case{Case $\mbox{\emph{Count}}(r)$} Count without grouping produces
  a relation with a single row. The stability of such a relation is 1, at any distance.

  \end{proof}

\vspace{-2mm}
\paragraph{Main theorem} We are now prepared to prove the main theorem.

\begin{theorem}
  The elastic sensitivity $\sfun^{(k)}(q,x)$ of a query $q$ at distance $k$ from the true database $x$ is an upper bound on the local sensitivity $A_{q}^{(k)}(x)$ of $q$ executed at distance $k$ from database $x$:
\vspace{-2mm}
 {\small \[ A_{q}^{(k)}(x) \leq \sfun^{(k)}(q,x) \] }
  \label{thm:local_sensitivity}
\end{theorem}

\rmspc \rmspc
\rmspc

\begin{proof}
  There are two cases: histogram queries and non-histogram queries.
  
  \case{Case $\mbox{\emph{Count}}(r)$ (non-histogram)} The local
  sensitivity of a non-histogram counting query over $r$ is equal to
  the stability of $r$, so the result follows directly from
  Lemma~\ref{sr_lemma}.

  \case{Case $\underset{G_1..G_n}{\mbox{\emph{Count}}}(r)$ (histogram)}
  In a histogram query, each changed row in the underlying relation can change two rows in the histogram~\cite{dwork2006calibrating}. Thus by Lemma~\ref{sr_lemma}, the histogram's local stability is bounded by $2\sfun^{(k)}_R(r, x)$.\qedhere
\end{proof}

\subsection{Optimization for Public Tables}
\label{sec:optim-joins-with}

\hyphenation{re-cords}
Our definition of elastic sensitivity assumes that all database records must be protected. In practice, databases often contain a mixture of sensitive and non-sensitive data. This fact can be used to tighten our bound on local sensitivity for queries joining on non-sensitive tables.

In our dataset,
for example, city data is publicly known, therefore the system 
does not need to protect against an attacker learning information
about the cities table. Note the set of public tables is domain-specific
and will vary in each data environment.

More precisely, in a join expression
\lstinline|T1 JOIN T2 ON T1.A = T2.B|, if \lstinline|T2| is publicly
known, the elastic stability of the join is equal to the elastic stability of
\lstinline|T1| times the maximum frequency of \lstinline|T2.B|. This formulation prevents the use of a
publicly-known table with repeated join keys from revealing information
about a private table.

%

\subsection{Discussion of Limitations and Extensions}
This section discusses limitations of elastic sensitivity and potential
extensions to support other common aggregation functions.

\subsubsection{Unsupported Queries}
\label{sec:unsupported-queries}
\Staticsensitivity does not support non-equijoins, and adding support
for these is not straightforward. Consider the query:

\begin{lstlisting}
SELECT count(*) FROM A JOIN B ON A.x > B.y
\end{lstlisting}

\noindent This query compares join keys using the greater-than
operator, and bounding the number of matches for this comparison would
require knowledge about \emph{all} the data for \lstinline|A.x| and
\lstinline|B.y|.

Fortunately, as demonstrated in our empirical study, more than three-quarters of
joins are equijoins. \Staticsensitivity could be extended to support 
other join types by querying the database for necessary data-dependent
bounds, but this modification would require interactions
with the database for each original query.

\Staticsensitivity can also fail when requisite max-frequency
metrics are not available due to the query structure. Consider the query:

\begin{lstlisting}
WITH A AS (SELECT count(*) FROM T1),
     B AS (SELECT count(*) FROM T2)
SELECT count(*) FROM A JOIN B ON A.count = B.count
\end{lstlisting}

\noindent This query uses counts computed in subqueries as join
keys. Because the $\mf$ metric covers only the attributes available in
the original tables of the database, our approach cannot bound the sensitivity
of this query and must reject it.
In general, \staticsensitivity applies only when join
keys are drawn directly from original tables. Fortunately, this criterion
holds for 98.5\% of joins in our dataset, so this limitation has
very little consequence in practice.

%
%
%
%
%
%
%
%
%
%
%
%
%

\subsubsection{Supporting Other Aggregation Functions}
\label{sec:other-aggregation-functions}
In this section we outline possible extensions of our approach to
support non-count aggregation functions, and characterize the expected utility for each.
These extensions, which provide a roadmap for potential future research, would
expand the set of queries supported by an elastic sensitivity-based system.

\paragraph{Value range metric}
To describe these extensions we define a new metric, \emph{value range} $\vi(a, r)$,
defined as the maximum value minus the minimum value allowed by the data model
of column $a$ in relation $r$.

This metric can be derived in a few ways. First, it can be extracted
automatically from the database's column constraint definitions~\cite{check_constraint},
if they exist. Second, a SQL query can extract the \emph{current} value range,
which can provide a guideline for selecting the permissible value range based on records
already in the database; finally, a domain expert can define the metric using knowledge about the
data's semantics.

Once the metric is defined, it must be enforced in order for differential
privacy to be guaranteed. The metric could be enforced as a data integrity check,
for example using column check constraints~\cite{check_constraint}.

\paragraph{Sum and Average}
For sum and average, we note that the local sensitivity of these functions
is affected both by the stability of the underlying relation,
because each row of the relation potentially contributes to the
computed sum or average, and by the range of possible values of the
attributes involved.

Given our definition of $\vi$ above, the elastic sensitivity of both
\texttt{Sum} and \texttt{Average} on relation $r$ at distance $k$ from database $x$ is defined by
$\vi(a, r) S^{(k)}_R(r, x)$. This expression captures the largest possible change in local sensitivity,
in which each new row in $r$ has the maximum value of $a$, for a total change of $\vi(a,r)$ per row.

For \emph{Sum} queries on relations with stability 1 (i.e. relations without joins),
this definition of \staticsensitivity is exactly equal to the query's local sensitivity,
so the approach will provide optimal utility.
As the relation's stability grows, so does the gap between \staticsensitivity
and local sensitivity, and utility degrades, since \staticsensitivity
makes the worst-case assumption that each row duplicated by a join contains the
maximum value allowed by the data model.

For the \emph{average} function, this definition is exactly equal to local sensitivity
only for the degenerate case of averages of a single row. As more input rows are added,
local sensitivity shrinks, since the impact of a single new row is amortized
over the number of averaged records, while elastic sensitivity remains constant.
Therefore utility degradation is proportional to both the stability of the relation as well
as the number of records being averaged.

This could be mitigated with a separate analysis
to compute a \emph{lower bound} on the number of records being averaged, in which case the
sensitivity could be scaled down by this factor. Such an analysis would require inspection of filter
conditions in the query and an expanded set of database metrics.

\paragraph{Max and min}
We observe that the stability of the underlying
relation has no effect on the local sensitivity of \emph{max} and
\emph{min}. Consequently, for such queries the data model $\vi(a,r)$ directly provides
the \emph{global sensitivity}, which is an
upper bound of local sensitivity. However, the max and min functions are inherently sensitive, because they are
strongly affected by outliers in the
database~\cite{dwork2009differential}, and therefore \emph{any} differential
privacy technique will provide poor utility in the general case.

Due to this fundamental limitation, previous work
~\cite{dwork2009differential,nissim2007smooth,smith2011privacy} has
focused on the \emph{robust} counterparts of these functions, such as
the interquartile range, which are less sensitive to changes in the
database.
This strategy is not viable in our setting since functions like
interquartile range are not supported by standard SQL.

\section{FLEX: Practical Differential \\Privacy for SQL Queries}
\label{sec:apply-smooth-sens}
\label{sec:flex}

\begin{figure}
  \centering
  \includegraphics[width=.48\textwidth]{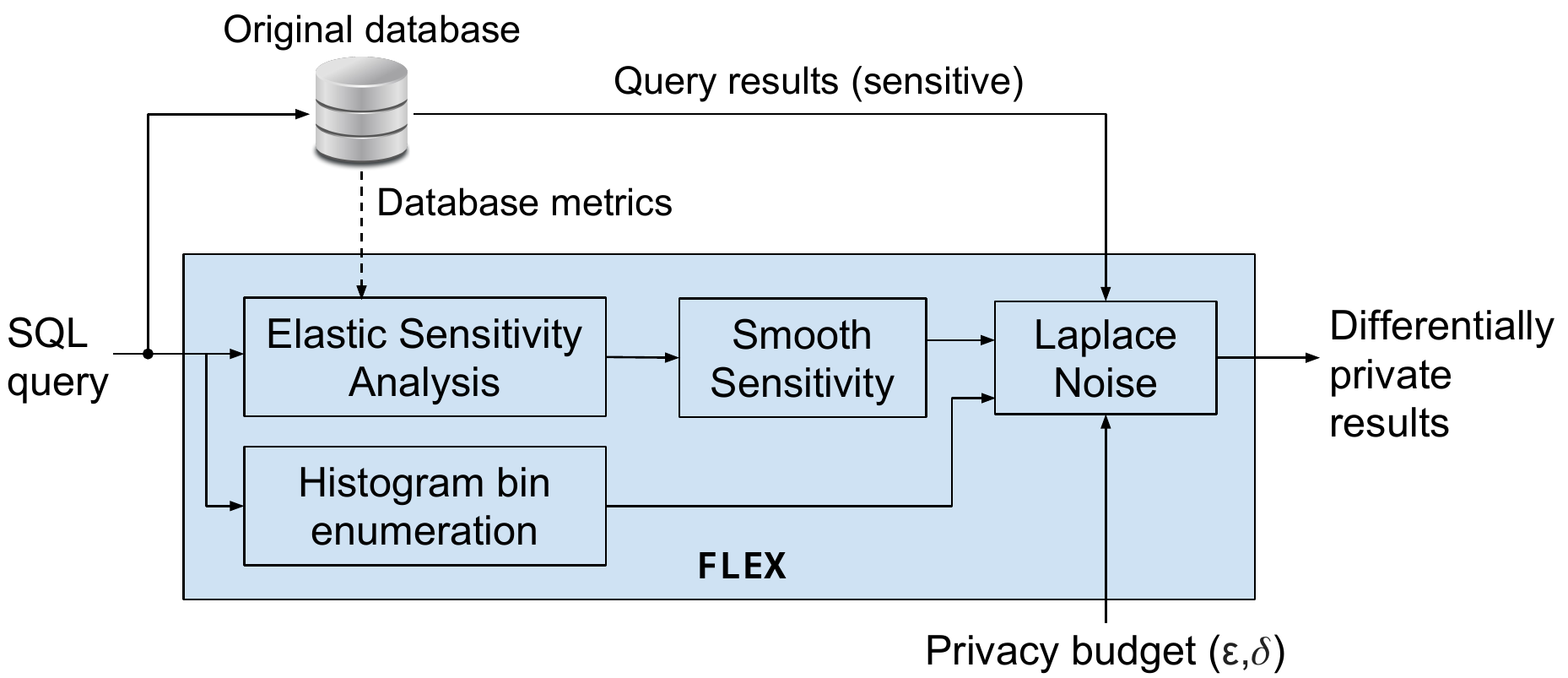}
  \caption{Architecture of \flex.}
  \label{fig:flex_arch}
\end{figure}

This section describes \flex, our system to enforce differential privacy
for SQL queries using \staticsensitivity. Figure~\ref{fig:flex_arch} summarizes
the architecture of our system. For a given SQL query, \flex uses an analysis
of the query to calculate its elastic sensitivity, as described in Section~\ref{sec:static-sensitivity}. \flex then applies smooth sensitivity~\cite{nissim2007smooth,nissim2011smooth} to the elastic sensitivity and finally adds noise drawn from the Laplace distribution to the original query results. In Section~\ref{sec:proof-correctness} we prove this approach provides $(\epsilon,\delta)$-differential privacy.

Importantly, our approach allows the query to execute on any existing database.
\flex requires only static analysis of the query and post-processing of the query results,
and requires no interactions with the database to enforce differential privacy.
As we demonstrate in Section~\ref{sec:exper-eval}, this design allows the approach to scale
to big data while incurring minimal performance overhead.

\paragraph{Collecting max frequency metrics}
The definition of elastic sensitivity requires a set of precomputed metrics $\mf$ from
the database, defined as the frequency of the most frequent attribute for each join key.
The values of $\mf$ can be easily obtained with a SQL query. For example, this query
retrieves the metric for column $a$ of table $T$:

\begin{lstlisting}[language=sql]
SELECT COUNT(a) FROM T GROUP BY a
ORDER BY count DESC LIMIT 1;
\end{lstlisting}

Obtaining these metrics is a separate step from enforcing differential privacy for a query; the metrics can be obtained once and re-used for all queries. Note the metric must be recomputed when the most frequent join attribute changes, otherwise differential privacy is no longer guaranteed. For this reason, the architecture in Figure~\ref{fig:flex_arch} is ideal for environments where database updates are far less frequent than database queries.

Most databases can be configured using triggers~\cite{triggers} to automatically recompute the metrics on database updates; this approach could support environments with frequent data updates.

\paragraph{Elastic Sensitivity analysis}
To compute elastic sensitivity we built an analysis framework for SQL queries
based on the Presto parser~\cite{presto}, with additional logic to resolve aliases
and a framework to perform abstract interpretation-based dataflow analyses on the query tree. 
\flex's elastic sensitivity analysis is built on this dataflow analysis engine, and propagates information about
ancestor relations and max-frequency metrics for each joined column in order to compute the
overall elastic sensitivity of the query, per the recursive definition in Section~\ref{sec:static-sensitivity}.
We evaluate the runtime and success rate of this analysis in Section~\ref{sec:exper-eval}.

\paragraph{Histogram bin enumeration}
When a query uses SQL's \lstinline|GROUP BY| construct, the output
is a histogram containing a set of bin labels and an aggregation result 
(e.g., count) for each bin.
To simplify presentation, our definition of elastic sensitivity in Section~\ref{sec:defin-stat}
assumes that the analyst provides the desired histogram bins labels $\ell$.
This requirement, also adopted
by previous work~\cite{mcsherry2009privacy}, is necessary to prevent leaking information via the presence or absence of a bin. 
In practice, however, analysts do not expect to provide histogram bin labels
manually.

In some cases, \flex can automatically build the set of histogram bin
labels $\ell$ for a given query.  In our dataset,
many histogram queries use non-protected bin labels drawn from 
finite domains (e.g. city names or product types). For each possible
value of the histogram bin label, \flex can automatically build $\ell$ and obtain the corresponding
differentially private count for that histogram bin. Then, \flex adds a
row to the output containing the bin label and its differentially
private count, where results for missing bins are assigned
value 0 and noise added as usual.

This process returns a histogram of the
expected form which does not reveal anything new through the
presence or absence of a bin. Additionally, since this process requires
the bin labels to be non-protected, the original bin labels can be
returned. The process can generalize to any aggregation function.

This process requires a finite, enumerable,
and non-protected set of values for each histogram bin label. %
When the requirement cannot be met, for example because
the histogram bin labels are protected or cannot be enumerated,
\flex can still return the differentially private count for each bin,
but it must rely on the analyst to specify the bin labels.

\subsection{Proof of Correctness}
\label{sec:proof-correctness}

In this section we formally define the \flex mechanism 
and prove that it provides $(\epsilon,\delta)$-differential privacy.

\flex implements the following differential privacy mechanism derived
from the Laplace-based smooth sensitivity mechanism defined by Nissim
et al.~\cite{nissim2007smooth,nissim2011smooth}:

\begin{definition}[\flex mechanism]
  For input query $q$ and histogram bin labels $\ell$ on true
  database $x$ of size $n$, with privacy parameters $(\epsilon, \delta)$:
  
  \begin{enumerate}
  \item Set $\beta = \frac{\epsilon}{2 \ln(2/\delta)}$.
  \item Calculate $S = \max_{k=0,1,...,n} e^{-\beta k} \sfun^{(k)}(q,x)$.
  \item Release $q_\ell(x) + \mbox{Lap}(2S/\epsilon)$.
  \end{enumerate}
  \label{def:flex}
\end{definition}

\noindent This mechanism leverages smooth
sensitivity~\cite{nissim2007smooth,nissim2011smooth}, using elastic sensitivity as an
upper bound on local sensitivity.

\begin{theorem} The \flex mechanism provides $(\epsilon,\delta)$-differential privacy.
  \label{thm:smooth_sensitivity}
\end{theorem}
\begin{proof}%
  By Theorem~\ref{thm:local_sensitivity} and Nissim et
  al.~\cite{nissim2011smooth} Lemma 2.3, $S$ is a $\beta$-smooth upper
  bound on the local sensitivity of $q$.
  By Nissim et al. Lemma 2.9, when the Laplace mechanism is used, a
  setting of $\beta = \frac{\epsilon}{2 \ln(2/\delta)}$ suffices to
  provide $(\epsilon, \delta)$-differential privacy.
  By Nissim et al. Corollary 2.4, the value released by the \flex
  mechanism is $(\epsilon, \delta)$-differentially private.
\end{proof}
\vspace*{2mm}

\subsection{Efficiently Calculating $S$}

The definition of the \flex mechanism (Definition~\ref{def:flex}) requires (in step 2) calculating the maximum smooth sensitivity over all distances $k$ between 0 and $n$ (the size of the true database). For large databases, this is inefficient, even if each sensitivity calculation is very fast.

The particular combination of elastic sensitivity with smooth sensitivity allows for an optimization. The elastic sensitivity $\sfun^{(k)}(q,x)$ grows as $k^{j(q)^2}$, where $j(q)$ is the number of joins in $q$ (see Lemma~\ref{lem:growth} below). For a given query, $j(q)$ is fixed. The smoothing factor ($e^{-\beta k}$), on the other hand, shrinks exponentially in $k$.

Recall that the smoothed-out elastic sensitivity at $k$ is $S(k) = e^{-\beta k} \sfun^{(k)}(q,x)$. We will show that to find $\max_{k=0,1,...,n} S(k)$, it is sufficient to find $\max_{k=0,1,...,m} S(k)$, where $m \geq \frac{j(q)^2}{\beta}$. Since $m$ is typically much smaller than $n$ (and depends on the query, rather than the size of the database), this observation yields significant computational savings.

\begin{lemma}
  For all relation expressions $r$ and databases $x$, where $j(r)$ is the number of joins in $r$, $\sfun_R^{(k)}(r,x)$ is a polynomial in $k$ of degree at most $j(r)^2$, and all coefficients are non-negative.
  \label{lem:growth}
\end{lemma}

\begin{proof}
  Follows from the recursive definitions of $\sfun_R^{(k)}$ and $\mf_k$, since each makes at most $j(r)$ recursive calls and only adds or multiplies the results.
\end{proof}

\begin{theorem}
  For all queries $q$ and databases $x$, the smoothed-out elastic sensitivity at distance $k$ is $S(k) = e^{-\beta k} \sfun^{(k)}(q,x)$. For each $x$ and $q$, if $q$ queries a relation $r$, the maximum value of $S(k)$ occurs from $k=0$ to $k = \frac{j(r)^2}{\beta}$.
\end{theorem}

\begin{proof}
  Let the constant $\lambda = j(r)^2$. By Lemma~\ref{lem:growth}, we have that for some set of constants $\alpha_0, \dots \alpha_{\lambda}$: 

  \[ \sfun_R^{(k)}(r,x) = \sum_{i = 0}^{\lambda} \alpha_i k^{\lambda -i} \]

  \noindent We therefore have that:

  \begin{align*}
    S(k) &= \frac{\sum_{i = 0}^{\lambda} \alpha_i k^{\lambda -i}}{e^{\beta k}}\\
    S'(k)&= {\sum_{i = 0}^{\lambda} \alpha_i e^{\beta k} k^{\lambda - i - 1} (\lambda - i - \beta k)}%
  \end{align*}

 Under the condition that $\alpha_i \geq 0$, each term in the numerator is $\leq 0$ exactly when $\lambda - i - \beta k \leq 0$. We know that $\alpha_i \geq 0$ by the definition of elastic sensitivity.
 
 We also know that $\lambda \geq 0$, because a query cannot have a negative number of joins.
 Thus the first term ($i = 0$) is $\leq 0$ exactly when $k \geq \frac{\lambda}{\beta}$ (we know that $\beta \geq 0$ by its definition). All of the other terms will also be $\leq 0$ when $k \geq \frac{\lambda}{\beta}$, because for $i > 0$, $\lambda - i - \beta k < \lambda - \beta k$.
 
 We can therefore conclude that $S'(k) \leq 0$ when $k > \frac{\lambda}{\beta}$, and so $S(k)$ is flat or decreasing for $k > \frac{j(r)^2}{\beta}$.
\end{proof}
 
\subsection{Privacy Budget \& Multiple Queries}

\flex does not prescribe a specific privacy budget management strategy,
allowing the use existing privacy budget methods as needed for specific applications.
Below we provide a brief overview of several approaches.

\paragraph{Composition techniques}
Composition for differential privacy~\cite{dwork2006calibrating} provides a simple way to support multiple queries: the $\epsilon$s and $\delta$s for these queries simply add up until they reach a maximum allowable budget, at which point the system refuses to answer new queries. The \emph{strong composition theorem}~\cite{dwork2010boosting} improves on this method to produce a tighter bound on the privacy budget used.
Both approaches are independent of the mechanism and thus apply directly to \flex.

\paragraph{Budget-efficient approaches}
Several approaches answer multiple queries together (i.e. in a single workload) resulting in more efficient use of a given privacy budget than simple composition techniques. These approaches work by posing counting queries through a low-level differentially private interface to the database. \flex can provide the low-level interface to support these approaches.

The \emph{sparse vector technique}~\cite{dwork2009complexity} answers only queries whose results lie above a predefined threshold. This approach depletes the privacy budget for answered queries only. The \emph{Matrix Mechanism}~\cite{li2010optimizing} and \emph{MWEM}~\cite{hardt2012simple} algorithms build an approximation of the true database using differentially private results from the underlying mechanism; the approximated database is then used to answer questions in the workload. Ding et al.~\cite{ding2011differentially} use a similar approach to release differentially private data cubes. Each of these mechanisms is defined in terms of the Laplace mechanism and thus can be implemented using \flex.

\section{Experimental Evaluation}
\label{sec:exper-eval}

\noindent We evaluate our approach with the following experiments:

\begin{itemize}[leftmargin=4mm]
\item We measure the performance overhead and success rate of \flex on
  real-world queries (Section~\ref{sec:allegro-eval}).

\item We investigate the utility of \flex-based differential privacy for real-world queries
  with and without joins (Section~\ref{sec:static-sensitivity-eval}).

\item We evaluate the effect of the privacy budget $\epsilon$ on the utility of
  \flex-based differential privacy (Section~\ref{sec:setting-parameters}).

\item We measure the utility impact of the public table optimization
  described in Section~\ref{sec:optim-joins-with} (Section~\ref{sec:impact_of_optimization}).

\item We compare \flex and wPINQ on a set of representative
counting queries using join (Section~\ref{sec:comp-with-wpinq}).

\end{itemize}

\paragraph{Experimental setup \& dataset} We ran all of our
experiments using our implementation of \flex with Java 8 on Mac
OSX. Our test machine was equipped with a 2.2 GHz Intel Core i7 and
8GB of memory. Our experiment dataset includes 9862 real queries executed
during October 2016. To build this dataset, we identified all counting queries
(including histogram queries) submitted during this time which
examined sensitive trip data. Our dataset also includes original 
results for each of these queries.

\begin{table}
  \centering
  
  \caption{Performance of \flex-based differential privacy.}
  \label{tab:eval-supported}

  \begin{tabular}{l r r}
    & \textbf{Avg (s)} & \textbf{Max (s)}\\
    \hline
    \emph{Original query}                      & 42.4  & 3,452 \\
    \emph{\flex: Elastic Sensitivity Analysis} & 0.007 & 1.2 \\
    \emph{\flex: Output Perturbation}          & 0.005 & 2.4 \\
    \hline
    \\
  \end{tabular}
\end{table}

\subsection{Success Rate and Performance of \flexlarge}
\label{sec:allegro-eval}

To investigate \flex's support for the wide 
range of SQL features in real-world queries,
we ran \flex's elastic sensitivity analysis on the queries
in our experiment dataset. We recorded the number of errors and classified each
error according to its type. %

In total, \flex successfully calculated elastic sensitivity for 76\% of the queries.
The largest group of errors is due to unsupported queries (14.14\%). These queries use features
for which our approach cannot compute an elastic sensitivity, as described in Section~\ref{sec:unsupported-queries}.
Parsing errors occurred for 6.58\% of queries. These errors result from incomplete grammar
definitions for the full set of SQL dialects used by the queries, and could be fixed
by expanding Presto parser's grammar definitions. The remaining errors (3.21\%) are due to 
other causes.

%
%

%
%
%
%
%
%
%

%
%
%
%

%

%
To investigate the performance of \flex-based differential privacy, we measured the total execution
time of the architecture described in Figure~\ref{fig:flex_arch} compared with the original query
execution time.
We report the results in Table~\ref{tab:eval-supported}.  Parsing and
analysis of the query to calculate elastic sensitivity took an average
of 7.03 milliseconds per query. The output perturbation step added an
additional 4.86 milliseconds per query.
By contrast, the average database execution time was 42.4 seconds per query, implying
an average performance overhead of 0.03\%.

\subsection{Utility of \flexlarge on Real-World Queries}
\label{sec:static-sensitivity-eval}

\newcommand{\targeted}{\emph{targeted\xspace}}
\newcommand{\nontargeted}{\emph{non-targeted\xspace}}

Our work is the first to evaluate differential privacy on a set of
real-world queries executed on real data. In contrast with previous evaluations of differential
privacy~\cite{DBLP:conf/ndss/BlockiDB16,DBLP:conf/sigmod/HayMMCZ16, DBLP:journals/pvldb/HuYYDCYGZ15},
our dataset includes a wide variety of real queries executed on real data.

We evaluate the behavior of \flex for this broad range of queries. Specifically, we measure the noise introduced to query results based on whether or not the query uses join and what percentage of the data is accessed by the query.

\paragraph{Query population size}
To evaluate the ability of \flex to handle both small and
large populations, we define a metric called \emph{population size}. The
population size of a query is the number of unique trips in the database
used to calculate the query results. The population size metric quantifies
the extent to which a query targets specific users or trips: a low
population size indicates the query is highly targeted, while a higher
population size means the query returns statistics over a larger
subgroup of records.

\begin{figure}%
  \centering \scriptsize
  \includegraphics[width=0.30\textwidth]{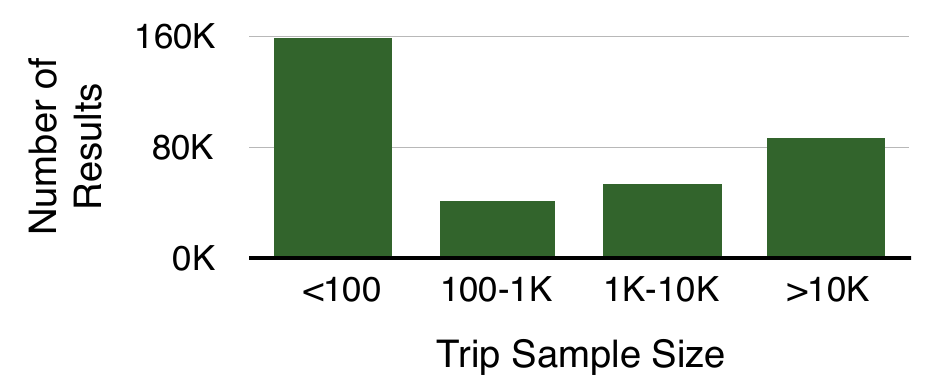}
  \caption{Distribution of population sizes for dataset queries.}
  \label{fig:sample_size_distribution}
\end{figure}

\ifvldb
\else
\begin{figure*}
  \centering
  \includegraphics[width=.95\textwidth]{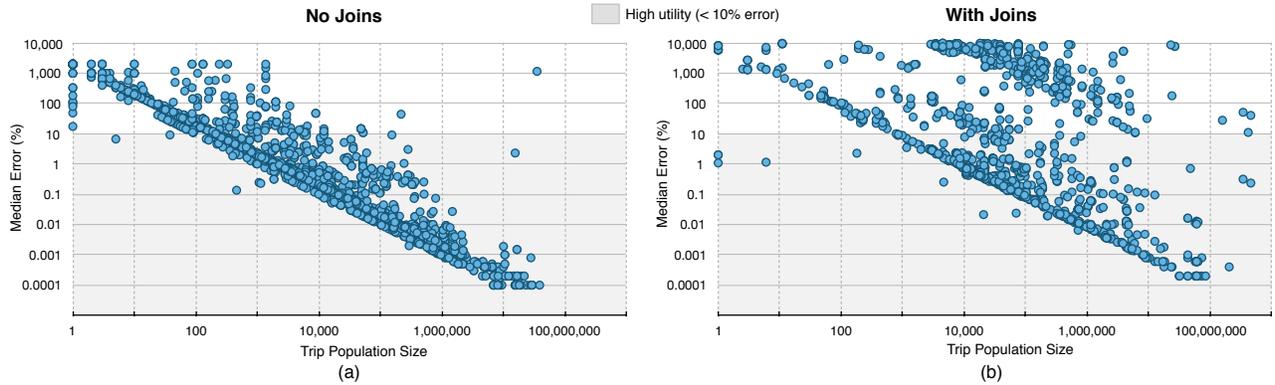}

  \caption{Median error vs population
    size for queries with no joins (a) and with joins (b).}
  \label{fig:sample_size_vs_noise}
\end{figure*}
\fi

Figure~\ref{fig:sample_size_distribution} summarizes the
distribution of population sizes of the queries in our dataset.
Our dataset contains queries with a wide variety of population sizes, reflecting the diversity of queries in the dataset.

\paragraph{Utility of \flex-based differential privacy}
We evaluate the utility of \flex by comparing the error
introduced by differential privacy on each query against the population size of
that query. For small population sizes, we expect our approach to protect
privacy by producing high error; for large population sizes, we expect our
approach to provide high utility by producing low error.

We used \flex to produce differentially private results
for each query in our dataset.
We report separately the results for queries with no joins and those with joins.
For each cell in the results, we calculated the relative (percent)
error introduced by \flex, as compared to the true
(non-private) results. Then, we calculated the median error of the
query by taking the median of the error values of all cells. For this
experiment, we set $\epsilon = 0.1$ and $\delta = n^{-\epsilon \ln
  n}$ (where $n$ is the size of the database), following Dwork and
Lei~\cite{dwork2009differential}.

Figure~\ref{fig:sample_size_vs_noise} shows the median error of each query against the population size of that query for queries with no joins (a) and with joins (b). The results indicate that \flex achieves its primary goal of supporting joins. Figure~\ref{fig:sample_size_vs_noise} shows similar trends with and without joins. In both cases the median error generally decreases with increasing population size; furthermore, the magnitude of the error is comparable for both. Overall, \flex provides high utility (less than 10\% error) for a majority of queries both with and without joins.

Figure~\ref{fig:sample_size_vs_noise}(b) shows a cluster of queries with higher errors but exhibiting the same error-population size correlation as the main group. The queries in this cluster perform many-to-many joins on private tables and do not benefit from the public table optimization described in Section~\ref{sec:optim-joins-with}. Even with this upward shift, a high utility is predicted for sufficiently large population size: at population sizes larger than 5 million the median error drops below 10\%.

Hay et al.~\cite{DBLP:conf/sigmod/HayMMCZ16} define the term \emph{scale-$\epsilon$ exchangeability} to describe the trend of decreasing error with increasing population size. The practical implication of this property is that a desired utility can always be obtained by using a sufficiently large population size.
For counting queries, a local sensitivity-based mechanism using Laplace noise is expected to exhibit scale-$\epsilon$ exchangeability. Our results provide empirical confirmation that \flex preserves this property, for both queries with and without joins, while calculating
an approximation of local sensitivity.

\begin{table}
  \centering
  \caption{Evaluated TPC-H queries.}
  \label{tbl:tpch_queries}
  \small
  \setlength\tabcolsep{1mm}
  \begin{tabular}{l l c}
    \textbf{Query}  & \textbf{Description} & \textbf{\# Joins} \\
    \hline
    ~~Q1         & Billed, shipped, and returned business & 0      \\
    ~~Q4         & Priority system status and customer satisfaction & 0    \\
    ~~Q13        & Relationship between customers and order size   & 1      \\
    ~~Q16        & Suppliers capable of supplying various part types & 1    \\
    ~~Q21        & Suppliers with late shipping times for required parts &  3     \\
    \hline
  \end{tabular}
\end{table}

\subsubsection{Utility of \flex on TPC-H benchmark}

We repeat our utility experiment using TPC-H~\cite{tpch},
an industry-standard SQL benchmark.
The source code and data for this experiment are available for download~\cite{experiment_code}.

The TPC-H benchmark includes synthetic data and queries simulating a workload for an archetypal industrial company.
The data is split across 8 tables (customers, orders, suppliers, etc.) and the
benchmark includes 22 SQL queries on these tables.

The TPC-H benchmark is useful for evaluating our system since the queries are specifically chosen to exhibit a high degree of complexity and to model typical business decisions~\cite{tpch}. This experiment measures the ability of our system to handle complex queries and provide high utility in a new domain.

\paragraph{Experiment setup}
We populated a database using the TPC-H data generation tool with the default scale factor of 1.
We selected the counting queries from the TPC-H query workload, resulting in five queries
for evaluation including three queries that use join.
The selected queries use SQL's GROUP BY operator and other SQL features including filters, order by, and subqueries. The selected queries are summarized in Table~\ref{tbl:tpch_queries}.
The remaining queries in the benchmark are not applicable for this experiment as they return raw data or use non-counting aggregation functions.

We computed the median population size and median error for each query using the same methodology as the previous experiment and privacy parameters $\epsilon=0.1$ and $\delta = n^{-\epsilon \ln n}$.
We marked as private every table containing customer or supplier information (customer, orders, lineitem, supplier, partsupp). The 3 tables containing non-sensitive metadata (region, nation, part) were marked as public.

\paragraph{Results}
The results are presented in Figure~\ref{fig:sample_size_vs_noise_tpch}. Elastic sensitivity
exhibits the same trend as the previous experiment: error decreases with increasing population size; this trend is observed for queries with and without joins, but error tends to be higher for queries with many joins.

\begin{figure}%
  \centering
  \includegraphics[width=0.47\textwidth]{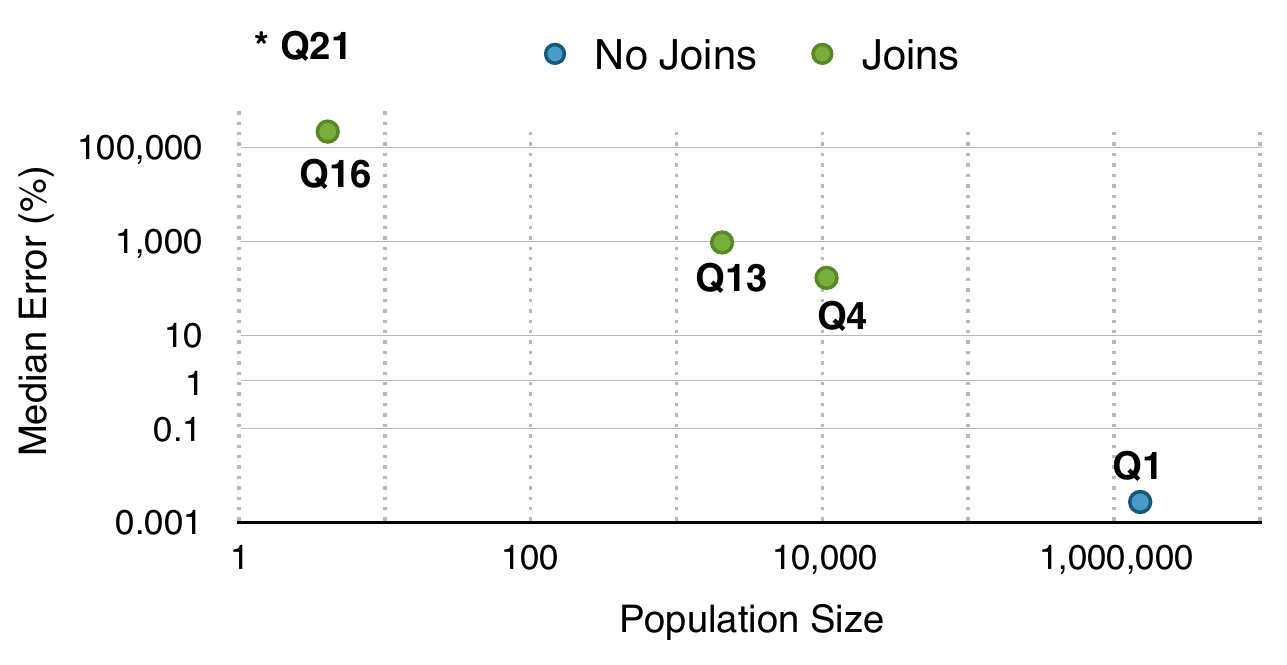}
  \caption{Median error vs population size (TPC-H queries).}
  \label{fig:sample_size_vs_noise_tpch}
\end{figure}

\subsubsection{Inherently sensitive queries}
Differential privacy is designed to provide good utility for
statistics about large populations in the data.
Queries with low population size, by definition, pose an inherent privacy risk to
individuals; differential privacy \emph{requires}
poor utility for their results in order to protect privacy.
As pointed out by 
Dwork and Roth~\cite{dwork2014algorithmic}, ``Questions about
specific individuals cannot be safely answered with accuracy, and
indeed one might wish to reject them out of hand.''

Since queries with low population size are inherently sensitive and
therefore not representative of the general class of queries of high
interest for differential privacy, we exclude queries with sample
size smaller than 100 in the remaining experiments. This ensures the
results reflect the behavior of \flex on queries for which high utility
may be expected.

\subsection{Effect of Privacy Budget}
\label{sec:setting-parameters}

In this section we evaluate the effect of the privacy budget on utility
of \flex-based differential privacy.
For each value of $\epsilon$ in the set $\{0.1, 1, 10\}$ (keeping
$\delta$ fixed at $n^{-\epsilon \ln n}$), we computed the median
error of each query, as in the previous experiment.

We report the results in Figure~\ref{fig:budget}, as a histogram
grouping queries by median error. As expected, larger values of
$\epsilon$ result in lower median error. When $\epsilon = 0.1$, \flex
produces less than 1\% median error for approximately half (49.8\%) of the less
sensitive queries in our dataset.

\begin{figure}%
  \centering
  \ifvldb
\else
    \includegraphics[width=.45\textwidth]{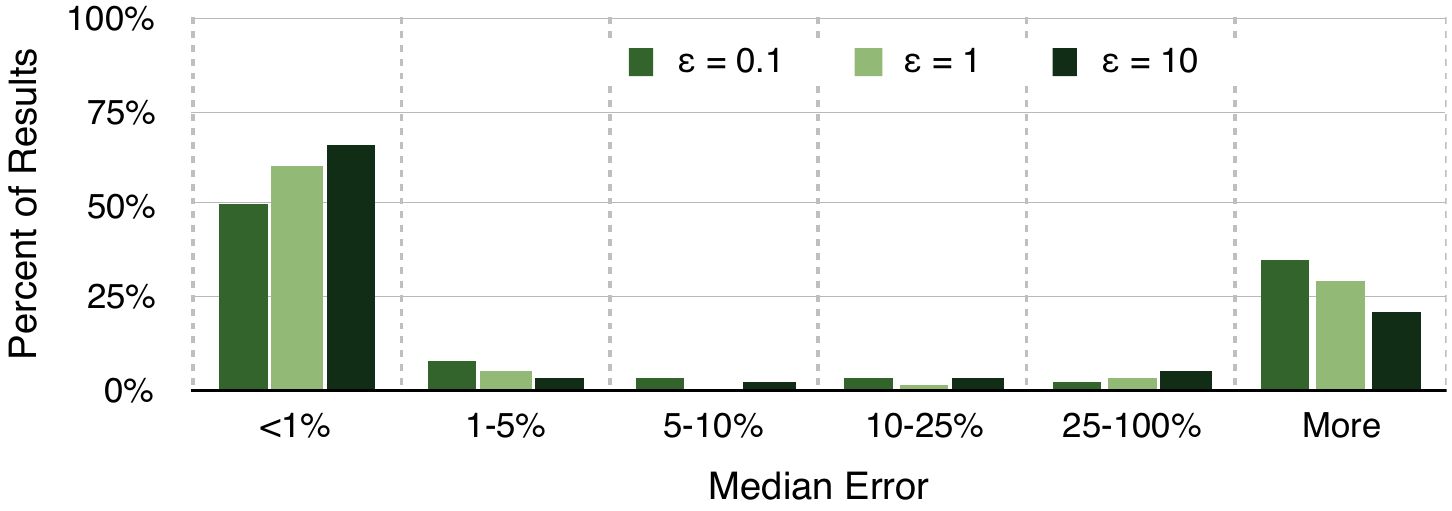}  
  \fi
  \caption{Effect of $\epsilon$ on median error.}
  \label{fig:budget}
\end{figure}

\paragraph{High-error queries}
The previous two experiments demonstrate that \flex
produces good utility for queries with high population size, but as
demonstrated by the number of queries in the ``More'' bin in
Figure~\ref{fig:budget}, \flex
also produces high error for some queries. 

To understand the root causes of this high error, we manually examined a random
sample of 50 of these queries and categorized them according to the
primary reason for the high error.

We summarize the results in
Table~\ref{fig:high-error-categories}.
The category \emph{filter on
  individual's data} (8\% of high error queries) includes queries that use a piece of
data specific to an individual---either to filter the sample with a
\texttt{Where} clause, or as a histogram bin.
For example, the
query might filter the set of trips by comparing the trip's driver ID
against a string literal containing a particular driver's ID, or it
might construct a histogram grouped by the driver ID, producing a
separate bin for each individual driver.
These queries are designed to return information specific to individuals.

The category \emph{low-population statistics} (72\% of high error queries)
contains queries with a
\texttt{Where} clause or histogram bin label that shrinks the set
of rows considered. A query to determine the success rate of a
promotion might restrict the trips considered to those within a small
city, during the past week, paid for using a particular type of credit
card, and using the promotion.
The analyst in this case may not intend
to examine the information of any individual, but since the query is
highly dependent on a small set of rows, the results may nevertheless
reveal an individual's information.

These categories suggest
that even queries with a population size larger than 100 can carry
inherent privacy risks, therefore differential privacy requires high
error for the reasons motivated earlier.

The third category (20\% of high error queries) contains queries that have
many-to-many joins with large maximum frequency metrics and which do not benefit
from any of the optimizations described in Section~\ref{sec:optim-joins-with}.
These queries are not necessarily inherently sensitive; the high error
may be due to a loose bound on local sensitivity arising from elastic 
sensitivity's design.

\begin{figure}
  \centering
  \ifvldb
\else
  \includegraphics[width=0.47\textwidth]{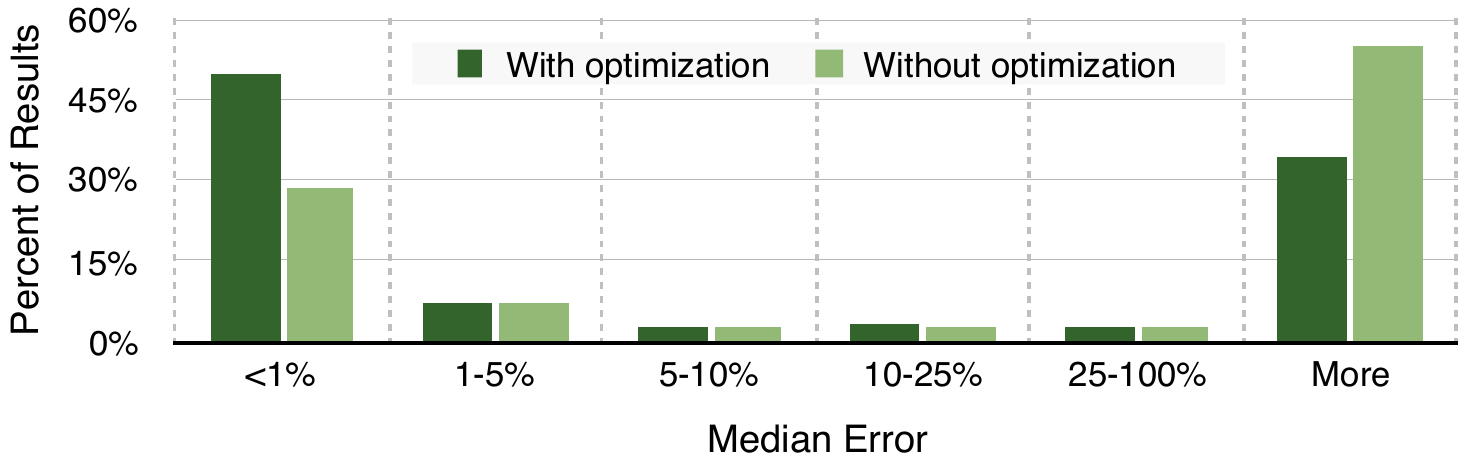}
  \fi
  \caption{Impact of public table optimization.}
  \label{fig:effect_of_optimization}
\end{figure}

\subsection{Impact of Public Table Optimization}
\label{sec:impact_of_optimization}

Section~\ref{sec:optim-joins-with} describes an optimization for queries joining on public tables. We measure the impact of this optimization on query utility by calculating median error introduced by \flex for each query in our dataset with the optimization enabled and disabled. We use the same experimental setup described in the previous section, with $\epsilon=0.1$ and $\delta = n^{-\epsilon \ln n}$. As before, we exclude queries with population size less than 100.

The optimization is applied to 23.4\% of  queries in our dataset.
Figure~\ref{fig:effect_of_optimization} shows the utility impact of the optimization across all queries. The optimization increases the percentage of queries with greatest utility (error less than 1.0\%) from 28.5\% to 49.8\%. The majority of the increase in high-utility queries come from the lowest-utility bin (error greater than 100\%) while little change is seen in the mid-range error bins. This suggests our optimization is most effective on queries which would otherwise produce high error, optimizing more than half of these queries into the 1\% error bin.

\ifvldb
\begin{table}
\else
\begin{table}[b]
\fi
\centering
\caption{Manual categorization of queries with high error.}
\label{fig:high-error-categories}
\small
  \begin{tabular}{l l}
    \textbf{Category} & \textbf{Percent} \\
    \hline
    Filters on individual's data & 8\% \\
    Low-population statistics & 72\% \\
    Many-to-many \texttt{Join} causes high \staticsensitivity & 20\% \\
    \hline
  \end{tabular}
  
\end{table}

\todo{make sure table names are consistent}

\newcommand{\proga}{%
Count distinct drivers who have completed a trip
in San Francisco yet enrolled as a driver
in a different city.
}
\newcommand{\tablesa}{trips, drivers}

\newcommand{\progb}{%
Count driver accounts that are active and were 
tagged after June 6 as duplicate accounts.
}
\newcommand{\tablesb}{users, user\_tags}

\newcommand{\progc}{%
Count motorbike drivers in Hanoi
who are currently active and have completed 10 or more trips.
}
\newcommand{\tablesc}{drivers, analytics}

\newcommand{\progd}{%
Histogram: Daily trips by city (for all cities) on Oct. 24, 2016.
}
\newcommand{\tablesd}{trips, cities}

\newcommand{\proge}{%
Histogram: Total trips per driver in Hong Kong between Sept. 9 and Oct. 3, 2016.
}
\newcommand{\tablese}{trips, drivers}

\newcommand{\progf}{%
Histogram: Drivers by different thresholds of total completed trips for drivers registered in Syndey, AUS who have completed a trip within the past 28 days.
}
\newcommand{\tablesf}{drivers, analytics}
 
\begin{table*}

  \centering
\caption{Utility comparison of wPINQ and \flex for selected set of representative counting queries using join.}
  \label{tbl:wpinq_comparison}
  \scriptsize

{\renewcommand{\arraystretch}{1.25}
\begin{tabular}{lp{10.5cm}lccc}
     & \multirow{3}{*}{\bf \hspace{-5mm}Program} & \multirow{3}{*}{\textbf{Joined tables}} &   \textbf{\scriptsize Median} & \multicolumn{2}{c}{\textbf{Median Error (\%)}}     \\
              &       &                 &    \textbf{\scriptsize Population}    & \multirow{2}{*}{wPINQ} & Elastic \\
              &       & & \textbf{\scriptsize Size} & & Sensitivity              \\
  \hline
     1. & \proga  &    \tablesa & 663 &  45.9   &  22.6 \\
     \hline
     2. & \progb  &    \tablesb & 734 &  71.5   &  ~2.8 \\
          \hline
     3. & \progc  &    \tablesc & 212 &  51.4   &  ~4.72 \\     
          \hline
     4. & \progd  &    \tablesd & 87  &  11.5   &  23 \\
          \hline
     5. & \proge  &    \tablese & 1   &  974~    & 6437 \\
          \hline
     6. & \progf  &    \tablesf & 72  &  51.5   &  27.8 \\     
  \hline
\end{tabular}}

\end{table*}

\subsection{Comparison with wPINQ}
\label{sec:comp-with-wpinq}

We aim to compare our approach to alternative differential privacy mechanisms with equivalent support for real-world queries. Of the mechanisms listed in Section~\ref{sec:exist-diff-priv}, only wPINQ supports counting queries with the full spectrum of join types.

Since wPINQ programs are implemented in C\#, we are unable to run wPINQ directly on our SQL query dataset. Instead we compare the utility between the two mechanisms for a selected set of representative queries. The precise behavior of each mechanism may differ for every query, however this experiment provides a relative comparison of the mechanisms for the most common cases.

\paragraph{Experiment Setup}
We selected a set of representative queries based on the most common equijoin patterns (joined tables and join condition) across all counting queries in our dataset. We identify the three most common join patterns for both histogram and non-histogram queries and select a random query representing each. Our six selected queries collectively represent 8.6\% of all join patterns in our dataset.

For each selected query we manually transcribe the query into a wPINQ program.
To ensure a fair comparison, we use wPINQ's \texttt{select} operator rather than the \texttt{join} operator for joins on a public table. This ensures that no noise is added to protect records in public tables, equivalent to the optimization described in Section~\ref{sec:optim-joins-with}.

Our input data for these programs includes all records from the cities table, which is public, and a random sample of 1.5 million records from each private table (it was not feasible to download the full tables, which contain over 2 billion records). We execute each program 100 times with the wPINQ runtime~\cite{wpinq_download}.

To obtain baseline (non-differentially private) results we run each SQL query on a database populated with only the sampled records. For elastic sensitivity we use max-frequency metrics calculated from this sampled data. We compute the median error for each query using the methodology described in the previous section, setting $\epsilon=0.1$ for both mechanisms. 

Table~\ref{tbl:wpinq_comparison} summarizes the queries and median error results. \flex provides lower median error than wPINQ for programs 1, 2, 3 and 6---more than 90\% lower for 2 and 3 and nearly 50\% lower for programs 1 and 6. \flex produces higher error than wPINQ for programs 4 and 5.

In the case of program 5, both mechanisms produce errors above 900\%. The median population size of 1 for this program indicates that our experiment data includes very few trips \emph{per driver} that satisfy the filter conditions. Elastic sensitivity provides looser bounds on local sensitivity for queries that filter more records, resulting in a comparably higher error for queries such as this one. Given that such queries are inherently sensitive, a high error (low utility) is required for \emph{any} differential privacy mechanism, therefore the comparably higher error of \flex is likely insignificant in practice.

Proserpio et al.~\cite{proserpio2014calibrating} describe a post-processing step for generating synthetic data by using wPINQ results to guide a Markov-Chain Monte Carlo simulation. The authors show that this step improves utility for graph triangle counting queries when the original query is executed on the synthetic dataset. 
While this strategy may produce higher utility than the results presented in Table~\ref{tbl:wpinq_comparison},
we do not evaluate wPINQ with this additional step since it is not currently automated.

\section{Related Work}
\label{sec:related-work}

Differential privacy was originally proposed by
Dwork~\cite{dworkdifferential2006, dwork2006calibrating,
  dwork2008differential}, and the reference by Dwork and
Roth~\cite{dwork2014algorithmic} provides an excellent general
overview of differential privacy.
Much of this work focuses on mechanisms for releasing the results of
specific algorithms. Our focus, in contrast, is on a general-purpose mechanism
for SQL queries that supports general equijoins.
We survey the existing general mechanisms that support join in
Section~\ref{sec:exist-diff-priv}.

 Lu et al.~\cite{lu2014generating} propose a mechanism for generating differentially private synthetic data such that queries with joins  have similar \emph{performance characteristics}, but not necessarily similar answers, on the synthetic and true databases.
However, Lu et al. do not propose a mechanism for answering queries with differential privacy. As such, it does not satisfy either 
of the two requirements in Section~\ref{sec:study-requirements}. 

Airavat~\cite{roy2010airavat} enforces differential privacy for
arbitrary MapReduce programs, but requires the analyst to bound the range of possible
outputs of the program, and clamps output values to lie within that
range.
Fuzz~\cite{gaboardi2013linear,haeberlen2011differential} enforces
differential privacy for functional programs, but does not support
one-to-many or many-to-many joins.

Propose-test-release~\cite{dwork2009differential} (PTR) is a framework
for leveraging local sensitivity that works for arbitrary real-valued
functions.
PTR requires (but does not define) a way to
calculate the local sensitivity of a function. Our work on elastic
sensitivity is complementary and can enable the use of PTR by
providing a bound on local sensitivity.

Sample \& aggregate~\cite{nissim2007smooth,nissim2011smooth} is a data-dependent
framework that applies to all statistical estimators. It works by
splitting the database into chunks, running the query on each chunk,
and aggregating the results using a differentially private
algorithm. Sample \& aggregate cannot support joins, since splitting
the database breaks join semantics, nor does it support queries that
are not statistical estimators, such as counting queries.
GUPT~\cite{mohan2012gupt} is a practical system that leverages the
sample \& aggregate framework to enforce differential privacy for
general-purpose analytics.

The \emph{Exponential Mechanism}~\cite{mcsherry2007mechanism} supports queries that produce categorical (rather than numeric) data. It works by randomly selecting from the possible outputs according to a \emph{scoring function} provided by the analyst.
Extending \flex to support the exponential mechanism would require specification of the scoring function and a means to bound its sensitivity.

A number of less-general mechanisms for performing specific graph
analysis tasks have been
proposed~\cite{hay2009accurate,sala2011sharing,karwa2011private,kasiviswanathan2013analyzing}.
These tasks often involve joins, but the mechanisms used to handle
them are specific to the task and are not applicable for
general-purpose analytics. For example, the recursive
mechanism~\cite{Chen:2013:RMT:2463676.2465304} supports general
equijoins in the context of graph analyses, but is restricted to
monotonic queries and in the worst case, runs in time exponential in
the number of participants in the database.

Kifer et al.~\cite{kifer2011no} point out that database constraints (such as uniqueness of a primary key) can lead to leaks of private data. Such constraints are common in practice, and raise concerns for \emph{all} differential privacy approaches. Kifer et al. propose increasing sensitivity based on the specific constraints involved, but calculating this sensitivity is computationally hard. Developing a tractable method to account for common constraints, such as primary key uniqueness, is an interesting target for future work.

%
%
%
%
%

%
%
%
%
%
%
%
%

%
%
%
%
%
%
%
%
%
%
%
%
%
%
%
%
%
%
%
%
%
%
%
%
%
%
%
%
%
%

%
%
%
%
%
%
%
%

 %
\section{Conclusion}
\label{sec:conclusion}

This paper takes a first step towards practical differential
privacy for general-purpose SQL queries.
To meet the requirements of real-world SQL queries, we
proposed \staticsensitivity, the first efficiently-computed
approximation of local sensitivity that supports joins.
We have released an open-source tool for computing elastic sensitivity of SQL queries~\cite{github_repo}.
We use \staticsensitivity to build \flex, a
system for enforcing differential privacy for SQL queries. We evaluated
\flex on a wide variety of queries,
demonstrating that \flex can support real-world queries
and provides high utility on a majority of queries with large population sizes.

\section*{Acknowledgments}

The authors would like to thank Abhradeep Guha Thakurta, Om Thakkar, Frank McSherry, Ilya Mironov and the anonymous reviewers for their
helpful comments.
This work was supported by the Center for Long-Term Cybersecurity, and DARPA \& SPAWAR under contract N66001-15-C-4066. The U.S. Government is authorized to reproduce and distribute reprints for Governmental purposes not withstanding any copyright notation thereon. The views, opinions, and/or findings expressed are those of the author(s) and should not be interpreted as representing the official views or policies of the Department of Defense or the U.S. Government.

\bibliographystyle{abbrv}
\bibliography{refs}

\end{document}